\documentclass[10pt]{article}
\usepackage[top=0.9in,bottom=1.1in,left=1.05in,right=1.05in]{geometry}
\usepackage{amsmath,amssymb}
\usepackage{amsthm}
\usepackage{latexsym}
\usepackage{mathrsfs,dsfont}
\usepackage{cancel}
\usepackage{comment}
\usepackage{float}
\usepackage{graphicx}
\usepackage{epstopdf}
\usepackage{color}
\usepackage{enumerate}
\usepackage{natbib}
\DeclareGraphicsExtensions{.eps}

\numberwithin{equation}{section}
\DeclareMathOperator*{\esssup}{ess\,sup}
\newcommand{\MCG}{\mathcal{G}}
\newcommand{\MCC}{\mathcal{C}}
\newcommand{\MCA}{\mathcal{A}}

\newcommand{\MCF}{\mathcal{F}}
\newcommand{\MCO}{\mathcal{O}}
\newcommand{\MCD}{\mathcal{D}}

\newcommand{\MCL}{\mathcal{L}}
\newcommand{\MCK}{\mathcal{K}}
\newcommand{\EE}{\mathbb{E}}
\newcommand{\PP}{\mathbb{P}}
\newcommand{\RR}{\mathbb{R}}

\newcommand{\Ltx}{\mathcal{L}_{t,x}}

\newcommand{\infnorm}[1]{\left\lVert#1\right\rVert_\infty}
\newcommand{\ltwonorm}[1]{\left\lVert#1\right\rVert_2}
\newcommand{\lfournorm}[1]{\left\lVert#1\right\rVert_4}

\newcommand{\Vy}{V^{\eps}}

\newcommand{\vz}{v^{(0)}}
\newcommand{\vo}{v^{(1)}}

\newcommand{\pz}{{\pi^{(0)}}}
\newcommand{\po}{{\pi^{(1)}}}

\newcommand{\Vzl}{V^{\pz,\delta}}

\newcommand{\Vyl}{V^{\pz,\eps}}
\newcommand{\Vyp}{V^{\pi,\eps}}

\newcommand{\pzt}{\widetilde\pi^0}
\newcommand{\pot}{\widetilde\pi^1}

\newcommand{\eps}{\epsilon}

\newcommand{\abs}[1]{\left|#1\right|}
\newcommand{\average}[1]{\left\langle#1\right\rangle}

\newcommand{\mc}[1]{\mathcal{#1}}

\newcommand{\ud}{\,\mathrm{d}}

\newcommand{\Wh}[1]{W^{(H)}_{#1}}

\newcommand{\half}{\frac{1}{2}}
\newcommand{\Yh}[1]{Y^{\eps,H}_{#1}}
\newcommand{\Yht}[1]{\widetilde Y^{\eps,H}_{#1}}
\newcommand{\kereps}{\MCK^\eps}

\newtheorem{theo}{Theorem}[section]
\newtheorem{lem}[theo]{Lemma}

\newtheorem{rem}[theo]{Remark}
\newtheorem{prop}[theo]{Proposition}
\newtheorem{assump}[theo]{Assumption}
\newtheorem{cor}[theo]{Corollary}

\begin{document}

\title{\vspace{-50pt} Optimal Portfolio under Fast Mean-reverting Fractional Stochastic Environment}
\author{Jean-Pierre Fouque\thanks{Department of Statistics \& Applied Probability,
 University of California,
        Santa Barbara, CA 93106-3110, {\em fouque@pstat.ucsb.edu}. Work  supported by NSF grant DMS-1409434.}
        \and Ruimeng Hu\thanks{Department of Statistics \& Applied Probability,
 University of California,
        Santa Barbara, CA 93106-3110, {\em hu@pstat.ucsb.edu}.}
        }
\date{\today}
\maketitle

\begin{abstract}
	
Empirical studies indicate the existence of long range dependence in the volatility of the underlying asset. This feature can be captured by modeling its return and volatility using functions of a stationary fractional Ornstein--Uhlenbeck (fOU) process with Hurst index $H \in (\half, 1)$. In this paper, we analyze the nonlinear optimal portfolio allocation problem under this model and in the regime where the fOU process is fast mean-reverting. We first consider the case of power utility, and rigorously give first order approximations of the value and the optimal strategy by a martingale distortion transformation. We also establish the asymptotic optimality in all admissible controls of a zeroth order trading strategy. Then, we consider the case with general utility functions using the epsilon-martingale decomposition technique, and we obtain similar asymptotic optimality results within a specific family of admissible strategies.

\end{abstract}


\section{Introduction}\label{sec_intro}
Asset allocation problems in continuous time framework are among the most widely studied problems in the field of mathematical finance, and has a long history dating back to \cite{Me:69, Me:71}. In his original work, explicit solutions are provided on how to trade stocks and/or to consume so that one's expected utility is maximized, when the underlying assets follow the Black--Scholes--Merton model, and where the utility functions are of specific type. Since these pioneering works, a large volume of research has been done for allowing financial market imperfections, for instance, see \cite{MaCo:76}, \cite{GuMu:13} for transaction costs,  \cite{GrZh:93},  \cite{CvKa:95},  \cite{ElTo:08} for investment under drawdown constraint, and \cite{CuCv:98} for trading with price impact.

Particularly, in the direction of asset modeling, the U-shaped pattern of Black--Scholes implied volatility from market option prices is widely observed when plotted against different strike prices, leading to the study of Merton problem when the volatility is stochastic, see \cite{Za:99}, \cite{ChVi:05}, \cite{FoSiZa:13} and \cite{LoSi:16}, to name a few. Moreover, empirical studies show that non-Markovian (dependence) structure models seem to better describe the data. Especially, in long-term investment which is related to daily data, long range dependence exhibits in both return and volatility: \cite{BrCrDe:98}, \cite{ChVi2:12, ChVi:12}, \cite{Co:01, Co:05}, \cite{EnPa:01}.

Our aim is to study the optimal portfolio problem when both return and volatility are driven by a long-range dependence process, denoted by $\Yh{t}$, which is \emph{fast-varying}. Specifically, we model $\Yh{t}$ by a stationary fractional Ornstein--Uhlenbeck process (fOU), which follows
\begin{equation*}
\ud \Yh{t} = -\frac{a}{\eps}\,\Yh{t} \ud t + \frac{1}{\eps^{H}}\ud \Wh{t}.
\end{equation*}
Here, $\eps$ is a small parameter to make the process $\Yh{t}$ fast-varying and its natural time scale to be of order $\eps$ (that is, its mean-reversion time scale proportional to $\eps$), and $\Wh{t}$ is a fractional Brownian motion (fBm) with Hurst index $H \in (\half,1)$ to give a fOU process that is of long-range dependence. A brief review regarding fBm and fOU is given in Section~\ref{sec_fBMfOU}, and the $\eps$-scaled fOU process $\Yh{t}$ is discussed in more details in Section~\ref{sec_fastfOU}. For further references, we refer to \cite{MaVa:68, ChKaMa:03,Co:07,BiHuOkZh:08,KaSa:11}.

The reason to consider such an asset modeling is threefold. 

Firstly, the stationary fOU process is Gaussian which makes its spectral decomposition (see \cite{FoPaSiSo:11}) available in explicit form when analyzing the properties of the Sharpe-ratio $\lambda(\Yh{t})$ introduced in Section \ref{sec_SVmerton}. The fOU process can be expressed as an integral of a well-studied kernel function with respect to a fBm process, which simplifies the derivation of needed estimates. Moreover, in addition to long-range correlation, it also satisfies other empirical ``stylized facts'', such as heavy tails and volatility clustering of returns, and persistence and mean-reversion of volatility as mentioned in \cite{Co:01,Co:05,EnPa:01}.

Secondly, when the process $\Yh{t}$ is slowly varying (that is $\eps$ large), which is particularly important in long-term investments, the asset allocation problem has been studied in \cite{FoHu:17} by a \emph{martingale distortion transformation} and regular perturbation techniques. So, it is natural to study the fast-varying  regime as well. 

Thirdly, although it is natural to consider multiscale factor  models for risky assets, with a slow factor and a fast factor as in \cite{FoSiZa:13} in a Markovian framework, the analysis requires more technical details, as the martingale distortion transformation is not available. This will be presented in another paper in preparation (\cite{Hu:XX}). 

In this paper, we focus on one-factor models and we 
study the effect of a fast time-scale on the optimal allocation problem.
The analysis of long-memory models is quite challenging. This is mainly due to the fact that the process $\Yh{t}$ is neither a semimartingale nor a Markov process. Consequently, the Hamilton-Jacobi-Bellman (HJB) partial differential equation (PDE) is not available, to which the singular perturbation technique is usually applied. Nevertheless, when the utility is of power type, a martingale distortion transformation is available and gives a representation of the value process as well as the optimal strategy. This result was originally discovered by \cite{Za:99} in the Markovian case and proved by applying a linearizing transformation to the HJB PDE. The general (non-Markovian) case was proved by \cite{Te:04} via a conditional H\"{o}lder inequality, and by \cite{FrSc:08} via a BSDE approach in the case of exponential utility. Recently, it has been revisited in \cite{FoHu:17} under the  setup \eqref{def_St} with a short proof based on a  verification argument. For general utilities, the problem can be investigated using the ``epsilon-martingale decomposition'' method. This approach was introduced in \cite{FoPaSi:00} and \cite{FoPaSi:01}, and recently developed in \cite{GaSo:15} for linear pricing problems where corrections to the Black--Scholes formula and implied volatility are derived when the fractional stochastic volatility  is slowly varying (or has small fluctuations), and in the fast mean-reverting regime in \cite{GaSo:16}.
  
\medskip
\noindent{\bf Main results.} In this paper, we study the nonlinear portfolio optimization problem under the fast-varying fractional stochastic environment described above. In the power utility case:
\begin{itemize}
\item
The value function and the optimal portfolio are obtained via the martingale distortion transformation stated in \cite{FoHu:17}. Using the ``ergodic property'' of $\Yh{t}$, we expand these expression in a probabilistic way, and deduce the first order approximations of both quantities. These approximations consist of a leading order term, which is related to the solution to the Merton problem with constant coefficients, and correction terms of order $\eps^{1-H}$. 

\item The asymptotics shares remarkable similarities with the slowly-varying case (see \cite{FoHu:17}): we find that without using  the correction term $\po$ of the optimal strategy, the leading order term $\pz$ by itself, generates the value process up to corrections of order $\eps^{1-H}$; and both $\pz$ and $\po$ are explicit in terms of the state processes (the wealth process $X_t$ defined below, $\Yh{t}$ and the time variable $t$). 

\item The later similarity, that is, $\pz$ and $\po$ are explicit in terms of $(X_t, \Yh{t}, t)$, however, leads to a non-trivial implementation of $\pz$, as the fast-varying $\Yh{t}$ needs to be tracked. We address this issue by suggesting  a sub-optimal practical (or lazy) strategy that sacrifices some accuracy in the value process.
\end{itemize}
For general utility functions, using the epsilon-martingale decomposition method and the properties of the risk tolerance function for the Merton problem with constant coefficients, we obtain an approximation for the portfolio value corresponding to a given strategy. As in \cite{FoHu:16} in the Markovian case, we show that this strategy is asymptotically optimal in a specific class of admissible strategies. 

The context of this paper is long-range correlation characterized by Hurst index $H \in (\half, 1)$.
As for the linear pricing problem in \cite{GaSo:16}, our results hold only in the range $H \in (\half, 1)$. The singular perturbation as $\eps\to 0$ does not commute with the limit $H \downarrow \half$ (see Section \ref{sec_comparisonMarkov}).  Therefore, the results in \cite{FoSiZa:13} in the Markovian case  can not be recovered by taking $H \downarrow \half$. The case $H \in (0, \half)$ corresponding to rough fractional stochastic volatility and short-term dependence is not addressed in this paper. Surprisingly, when $H < 1/2$, the first order corrections to the value process appear to be of order $\sqrt{\eps}$, and $\Yh{t}$ is not visible to the leading order nor in the corrections. These findings will be presented in another paper in preparation (\cite{FoHu:18}). In fact, the proofs of crucial lemmas in Appendix \ref{app_lemmas} break down when $H < \half$, which is translated into divergent integrals in that case.

\medskip
\noindent{\bf Organization of the paper.} The rest of the paper is organized as follows. In Section~\ref{sec_SVmerton}, we restate the martingale distortion transformation under general stochastic volatility models. This is derived in the Markovian case in \cite{Za:99}, and in non-Markovian settings in \cite{Te:04,FrSc:08,FoHu:17}. We also review the fractional Brownian motion and fractional Ornstein--Uhlenbeck processes. In Section~\ref{sec_asymppower}, we introduce the fast-varying long-range dependence stochastic factor $\Yh{t}$ modeled by the $\eps$-scaled fOU process. Then, the asymptotic results under this modeling are derived and given in Section~\ref{sec_asympVy} and \ref{sec_asymppi} for the value process and optimal portfolio respectively. Asymptotic optimality of the leading order strategy $\pz$ in the full class of admissible strategies up to $\eps^{1-H}$ is discussed, and the implementation difficulties are also addressed with numerical illustrations. We also compare the results with the Markovian case and comment on the influence of long-range dependence models. The problem with general utility functions is discussed and similar asymptotic optimality results are presented in Section~\ref{sec_optimality}. We make conclusive remarks in Section~\ref{sec_conclusion}.

\section{Merton problem under one factor stochastic environment  and power utility}\label{sec_SVmerton}

Let $S_t$ be the price of the underlying asset at time $t$, whose return and volatility are driven by a stochastic factor $Y_t$, 
\begin{align}
\ud S_t = S_t\left[\mu(Y_t)\ud t  + \sigma(Y_t)\ud W_t\right]. \label{def_St}
\end{align}
Here $Y_t$ is a general stochastic process adapted to the natural filtration $\MCG_t$ generated by a Brownian motion $W_t^Y$  correlated with the Brownian motion $W_t$ which drives the asset price $S_t$:
\begin{equation}\label{def_BMcorrelation}
\ud \average{W_t, W_t^Y} = \rho \ud t, \quad \abs{\rho}<1.
\end{equation}
We also define $\MCF_t$ as the natural filtration generated by the two Brownian motions $(W_t, W_t^Y)$.

Let $\pi_t$ be the amount of money invested in the underlying asset $S_t$ at time $t$, 
while the rest earns a constant interest rate $r$. We require $\pi_t$ to be $\MCF_t$-adapted and self-financing. Denote by $X_t^\pi$ the wealth process associated to the strategy $\pi$,  and, without loss of generality, assume that the interest rate $r$ is zero, then the dynamics of $X_t^\pi$ is given by:

\begin{equation}
\ud X_t^\pi = \pi_t\mu(Y_t) \ud t + \pi_t\sigma(Y_t) \ud W_t.\label{def_XtunderY}
\end{equation}
The investor aims at finding the optimal strategy in order to maximize her expected utility of terminal wealth $X_T^\pi$. Mathematically, it consists in  identifying the value process $V_t$ defined by
\begin{equation}\label{def_Vt}
V_t := \esssup_{\pi \in \MCA_t}\EE\left[U(X_T^\pi)\vert \MCF_t\right],
\end{equation}
and the corresponding optimal strategy $\pi^\ast$, given the investor's utility function $U(\cdot)$. The form of $U(\cdot)$ varies from section to section. Specifically, in the rest of this section and Section~\ref{sec_asymppower}, we will work with power utilities under Assumption~\ref{assump_power}, namely 
\begin{equation}\label{def_power}
U(x) = \frac{x^{1-\gamma}}{1-\gamma}, \quad \gamma >0, \quad \gamma \neq 1,
\end{equation}
while in Section~\ref{sec_optimality}, the utility function is in general form satisfying Assumption~\ref{assump_U}. The set $\MCA_t$ contains all admissible strategies: 
\begin{equation}\label{def_MCA}
\MCA_t := \left\{\pi \text{ is } (\MCF_t)\text{-adapted} : X_s^\pi \text{ in } \eqref{def_XtunderY} \text{ stays nonnegative } \forall s \geq t, \text{ given } \MCF_t\right\},
\end{equation}
with zero being an absorbing state for $X^\pi$ (bankruptcy). Additionally, for the power utility case, we require that for all $\pi \in \MCA_t$, the following integrability conditions are satisfied: 
	\begin{align}\label{assump_strategies}
	\sup_{t\in[0,T]}\EE\left[ \left(X_t^\pi\right)^{2p(1-\gamma)}\right] < +\infty,\; for \; some  \;p > 1, \quad and \quad \EE \left[\int_0^T\left(X_t^\pi\right)^{-2\gamma}    \pi_t^2  \sigma^2(Y_t)\ud t \right] < \infty.
	\end{align}

In \cite{FoHu:17}, the value process \eqref{def_Vt} is studied when $U(x)$ is of power type and is represented via a \emph{martingale distortion transformation}. For readers' convenience,  we first briefly review this representation. Then, as a preparation for working under a specific fractional stochastic environment, we review the fractional Brownian motion (fBm) and fractional Ornstein-Uhlenbeck (fOU) processes.
 

\subsection{Martingale distortion transformation}\label{sec_martdistort}
The martingale distortion transformation was derived in \cite{Te:04} with a slightly different utility function, and recently stated in \cite{FoHu:17} under the same setup as in this paper.

Denote by $\widetilde \PP$ the probability measure  defined by 
\begin{equation}\label{def_Ptilde}
\frac{\ud \widetilde \PP}{\ud \PP} = \exp\left\{-\int_0^T a_s \ud W_s^Y - \half \int_0^T a_s^2 \ud s \right\},
\end{equation}
where 
\begin{equation}\label{def_at}
a_t = -\rho\left(\frac{1-\gamma}{\gamma}\right)\lambda(Y_t),
\end{equation}
is bounded and $\MCG_t$-adapted.
Therefore, $\widetilde W_t^Y :=  W_t^Y+ \int_0^t a_s \ud s$ is a $\widetilde\PP$-Brownian motion.


\begin{assump}\label{assump_power}\quad
	\begin{enumerate}[(i)]
		\item\label{assump_St} The SDE \eqref{def_St} for $S_t$ has a unique strong solution, in other words, $$S_t=S_0e^{\int_0^t(\mu(Y_s)-\half\sigma^2(Y_s))\ud s+\int_0^t\sigma(Y_s)\ud W_s}$$
		exists for all $t\in [0,T]$. 
		\item\label{assump_filtration} Assume the filtration generated by $(Y_s)_{s \leq t}$ is also $\MCG_t$, and the volatility function $\sigma(\cdot)$ is injective.
		
		\item\label{assump_lambda} The Sharpe ratio $\lambda(\cdot) := \mu(\cdot)/\sigma(\cdot)$ is assumed to be bounded and $C^2(\RR)$. Also, the derivatives $\lambda'$ and $\lambda''$ are assumed bounded. 
		\item\label{assump_xi} Define the $\widetilde \PP$-martingale
		\begin{equation}\label{def_Mtmartdistort}
		M_t = \widetilde \EE\left[\left.e^{\frac{1-\gamma}{2q\gamma}\int_0^T \lambda^2(Y_s)\ud s} \right\vert \MCG_t\right],
		\end{equation}
		and write its representation
		\begin{equation}\label{def_xi}
		\ud M_t = M_t\xi_t \ud \widetilde W_t^Y.
		\end{equation}
		We assume
		\begin{equation*}
		\EE\left[e^{c_\xi\int_0^T \xi_t^2\ud t}\right] < \infty, 
		\end{equation*}
		where the constant $c_\xi$ is given by 
		$c_\xi = \frac{16 (1-\gamma)^2\rho^2 p^2q^2}{\gamma^2}$ for  $\gamma < 1$, and $c_\xi = \frac{16 (1-\gamma)^2\rho^2 p^2q^2}{\gamma^2} - \frac{4p(1-\gamma)}{\gamma^2}$ for  $\gamma >1$. The parameter $p$ is introduced in \eqref{assump_strategies} and $q$ is defined in terms of $\gamma$ and $\rho$ by
				\begin{equation}\label{def_q}
				q=\frac{\gamma}{\gamma+(1-\gamma)\rho^2}.
				\end{equation}
		Note that $q$ is the usual ``distortion'' exponent firstly introduced in \cite{Za:99}.	
	\end{enumerate}
\end{assump}

\begin{rem}
In the above assumptions, by part \eqref{assump_filtration}, $\MCF_t$ is also the filtration generated by $(W, Y)$. Since $\sigma(\cdot)$ is one-to-one, it is also  the one generated by $S_t$. This assumption is important, since in reality $S_t$ is what we observed. In part \eqref{assump_lambda}, the smoothness of $\lambda$ is needed when Taylor expansions are performed to prove Theorem~\ref{thm_Vtpowerexpansion} and \ref{thm_piexpansion}; while the boundedness of $\lambda'$ and $\lambda''$ are convenient assumptions for the estimation of the error terms. Part \eqref{assump_xi} on $(\xi)_{t\in[0,T]}$ seems to be a strong assumption. However, in Section~\ref{sec_asymppower}, we will see that it is satisfied for our proposed model for $\Yh{t}$.
	
The assumption that $r=0$ can be viewed as a change of num\'{e}raire. In fact, the following Proposition \ref{prop_martdistort} could be extended to the case $r = r(Y_t)$ with only minor modifications. 
\end{rem}

\begin{prop}[Martingale Distortion Transformation]\label{prop_martdistort}
	Let $S_t$ follow the dynamics \eqref{def_St}, and suppose the objective is \eqref{def_Vt} with the power utility function \eqref{def_power}. Under Assumptions~\ref{assump_power}, the value process $V_t$ is given by
	\begin{equation}\label{def:Vcorrelated}
	V_t=\frac{X_t^{1-\gamma}}{1-\gamma} \left[\widetilde{\EE}\left(\left.e^{\frac{1-\gamma}{2q\gamma}\int_t^T\lambda^2(Y_s)\ud s}\right\vert \mc{G}_t\right)\right]^q.
	\end{equation}
	The expectation
	$\widetilde \EE[\cdot]$ is computed with respect to $\widetilde{\PP}$ introduced in \eqref{def_Ptilde}.
	The parameter $q$ is given by \eqref{def_q}.
	
	The optimal strategy $\pi^\ast$ is
	\begin{equation}\label{def_pioptimal}
	\pi^\ast_t = \left[\frac{\lambda(Y_t)}{\gamma \sigma(Y_t)} + \frac{\rho q \xi_t}{\gamma \sigma(Y_t)}\right] X_t,
	\end{equation}
	where $\xi_t$ is given by the Martingale Representation Theorem in \eqref{def_xi}.
\end{prop}
\begin{proof}
	See \cite[Proposition~2.2]{FoHu:17} for a detailed proof. 
\end{proof}

\begin{rem}
The separation of variable form \eqref{def:Vcorrelated}, that is, the utility of the current wealth $U(X_t)$ multiplied by a process related to the stochastic factor $Y_t$, is motivated by the Markovian case firstly developed in \cite{Za:99}. 

When $Y_t$ is Markovian, results in \cite{Za:99} is recovered by rewriting the value process:
$$V_t =  U(X_t)\left[\widetilde{\EE}\left(\left.e^{\frac{1-\gamma}{2q\gamma}\int_t^T\lambda^2(Y_s)\ud s}\right\vert Y_t\right)\right]^q = U(X_t) v(Y_t)^q,$$
and applying the Feyman-Kac formula to $v(\cdot)$.

When the Sharpe-ratio is degenerate $\lambda(y) = \lambda_0$, the value process and the optimal strategy are reduced to
$$
V_t = \frac{X_t^{1-\gamma}}{1-\gamma}e^{\frac{1-\gamma}{2\gamma}\lambda_0^2(T-t)}, \quad \pi^\ast_t = \frac{\lambda_0}{\gamma \sigma(Y_t)}X_t.
$$

When the two Brownian motions $W_t$ and $W_t^Y$ are uncorrelated $\rho=0$, the problem is already ``linear'' since $q = 1$. In that case, the value process and the optimal strategy are simplified as:
$$
	V_t=\frac{X_t^{1-\gamma}}{1-\gamma} {\EE}\left[\left.e^{\frac{1-\gamma}{2q\gamma}\int_t^T\lambda^2(Y_s)\ud s}\right\vert \mc{G}_t\right], \quad \pi_t^\ast = \frac{\lambda(Y_t)}{\gamma \sigma(Y_t)}X_t.
$$

The results in Proposition~\ref{prop_martdistort} can also be generalized to the case of log utility and/or with multiple assets, see \cite{FoHu:17} for further discussion.
\end{rem}

\subsection{Fractional Brownian motion and fractional Ornstein-Uhlenbeck processes}\label{sec_fBMfOU}

A standard fractional Brownian motion (fBm) is a continuous Gaussian process $\left(\Wh{t}
\right)_{t\in\RR}$ with zero mean and  covariance structure:
\begin{equation*}
\EE\left[\Wh{t}\Wh{s}\right] = \frac{\sigma_H^2}{2}\left(\abs{t}^{2H} + \abs{s}^{2H} - \abs{t-s}^{2H}\right),
\end{equation*}
where $\sigma_H$ is a positive constant 
\begin{equation}\label{eq_Whvar}
\sigma^2_H = \frac{1}{\Gamma(2H+1)\sin(\pi H)},
\end{equation}
and $H \in (0,1)$ is the Hurst index. According to \cite{MaVa:68}, $\Wh{t}$ can be represented by the following moving-average integral:
\begin{equation}\label{eq_movingfBM}
\Wh{t} = \frac{1}{\Gamma(H+\half)} \int_{\RR} \left((t-s)_+^{H-\half} - (-s)_+^{H-\half}\right) \ud W_s^Y,
\end{equation}
where $(W_t^Y)_{t\in\RR^+}$ is the standard Brownian motion that is correlated with $W_t$ as given in \eqref{def_BMcorrelation}, and $(W_t^Y)_{t\in\RR^-} := \left(B_{-t}\right)_{t\in\RR^-}$ is another Brownian motion independent of $(W_t^Y)_{t\in\RR^+}$ and $(W_t)$. 

We then introduce the stationary fractional Ornstein-Uhlenbeck (fOU) process as
\begin{equation}\label{eq_fOUsol}
Y_t^H := \int_{-\infty}^t e^{-a(t-s)} \ud \Wh{s} 
\end{equation}
which is the unique (in distribution) stationary solution to the Langevin equation driven by fBm (see \cite{ChKaMa:03})
\begin{equation}\label{eq_fOU}
\ud Y_t^H = -aY_t^H \ud t + \ud \Wh{t}, 
\end{equation}
where $a>0$ is a strictly positive parameter. It has zero mean and (co)variance structure:
\begin{align}
&\sigma_{ou}^2:= \EE\left[\left(Y_t^H\right)^2\right]  = \half a^{-2H}\Gamma(2H+1)\sigma_H^2,\label{eq_Zhvar}\\
&\EE\left[Y_t^H Y_{t+s}^H\right] = \sigma_{ou}^2 \frac{2\sin(\pi H)}{\pi} \int_0^\infty \cos(asx)\frac{x^{1-2H}}{1+x^2}\ud x:= \sigma^2_{ou}\MCC_Y(s).\label{eq_Zhcovar}
\end{align}
By the moving-average representation \eqref{eq_movingfBM} for $\Wh{t}$, the stationary solution \eqref{eq_fOUsol} is expressed as:
\begin{equation}\label{eq_fOUker}
Y_t^H = \int_{-\infty}^t \mc{K}(t-s)\ud W_s^Y,
\end{equation}
where $\left(W_t^Y\right)_{t\in\RR}$ is the standard Brownian motion on $\RR$ as described after equation \eqref{eq_movingfBM}.
The non-negative kernel $\mc{K}$ takes the form
\begin{equation}\label{def_kernel}
\mc{K}(t) = \frac{1}{\Gamma(H+\half)} \left[t^{H-\half} - a \int_0^t (t-s)^{H-\half}e^{-as}\ud s\right],
\end{equation}
and $\int_0^\infty \MCK^2(u) \ud u = \sigma^2_{ou}$. For asymptotic properties of $\MCK(t)$ when $t \ll 1$ and $t \gg 1$, we refer to \cite[Section 2.2]{GaSo:15}. They also provide short-range correlation properties when $H \in (0,\half)$, and long-range correlation properties when $H \in (\half,1)$. In Section~\ref{sec_asymppower}, we will mainly focus on the case of $H > \half$, as explained in the introduction. Specifically, we will study the Merton problem \eqref{def_Vt} when $Y_t$ follows a rescaled version of \eqref{eq_fOUsol}, such that it is fast-varying. 

\section{Application to fast-varying fractional stochastic environment}\label{sec_asymppower}

In this section, we first introduce the $\eps$-scaled stationary fOU process denote by $\Yh{t}$, of which we mention several properties with proofs delayed to the Appendix. Then, we  study the Merton problem \eqref{def_Vt} under such fractional stochastic factor $\Yh{t}$. To be specific, we will give approximations of both the value process, denoted by $\Vy_t$ and the corresponding optimal strategy $\pi^\ast$. This is done by applying Proposition~\ref{prop_martdistort} with $Y_t = \Yh{t}$, then by expanding the expressions \eqref{def_Vy} based on the properties mentioned in Section~\ref{sec_fastfOU}. We also show that the ``leading order'' strategy alone can produce the given approximation of $\Vy_t$. 
However, the implementation needs to track the fast factor $\Yh{t}$ using high-frequency data and this is not an easy task. To address this issue, we propose a practical strategy which does not require tracking $\Yh{t}$, with numerical illustration. Finally, we compare the results with the Markovian case, and we comment on the effects of taking into account the long-range dependence. 

\subsection{The fast mean-reverting fOU process}\label{sec_fastfOU}

The $\eps$-scaled fractional Ornstein--Uhlenbeck process $\Yh{t}$ is defined by
\begin{align}\label{def_Yh}
\Yh{t} := \eps^{-H} \int_{-\infty}^t e^{-\frac{a(t-s)}{\eps}} \ud \Wh{s}
\end{align}
where $\eps \ll 1$ is a small parameter and $H \in(\half,1)$. In the following moving-average integral representation
\begin{equation}
\Yh{t}= \int_{-\infty}^t \kereps(t-s) \ud W_s^Y, \quad \kereps(t) = \frac{1}{\sqrt{\eps}}\mc{K}\left(\frac{t}{\eps}\right),\label{eq_Yh}
\end{equation}
 $W^Y$ is the Brownian motion that drives the process $\Yh{t}$ as in \eqref{eq_fOUker}, and is correlated with $W_t$ as in \eqref{def_BMcorrelation}.

It is a zero-mean, stationary Gaussian process with variance $\sigma_{ou}^2$ and covariance
\begin{equation}
\EE\left[\Yh{t}\Yh{t+s}\right] =\sigma^2_{ou}\MCC_Y\left(\frac{s}{\eps}\right) = \sigma_{ou}^2 \frac{2\sin(\pi H)}{\pi} \int_0^\infty \cos\left( \frac{asx}{\eps}\right)\frac{x^{1-2H}}{1+x^2}\ud x,
\end{equation}
which shows the natural scale of $\Yh{t}$ is $\eps$ as desired. Moreover, the correlation function $C_Y(s)$ is not integrable at infinity and the long-range correlation exhibits the behavior:
\begin{equation*}
\MCC_Y(s) = \frac{(as)^{2H-2}}{\Gamma(2H-1)} + o(s^{2H-2}), \quad s \gg 1.
\end{equation*}
The Sharpe-ratio process $\lambda(\Yh{t})$ inherits this long-range correlation, namely,
\begin{equation*}
Cov(\lambda(\Yh{t}), \lambda(\Yh{t+s})) = Var(\lambda^2(\Yh{t}))C_\lambda(\frac{s}{\eps}), \text{ and } C_\lambda(s) \sim \MCO(s^{2H-2}), \text{ for } s \gg 1.
\end{equation*}
This follows from a straightforward modification of proofs in \cite[Lemma~3.1]{GaSo:16}.

Now, we check that Assumption~\ref{assump_power}\eqref{assump_xi} is satisfied by $\Yh{t}$.

\begin{lem}
	Under Assumption~\ref{assump_power}\eqref{assump_St}-\eqref{assump_lambda}, the fast mean-reverting stationary fractional Ornstein--Uhlenbeck process $\Yh{t}$ defined in \eqref{def_Yh} satisfies  Assumption~\ref{assump_power}\eqref{assump_xi}.
\end{lem}
\begin{proof}
	This is a slightly different version of Lemma~3.1 in \cite{FoHu:17}. Using the property that $A^\eps(T) \equiv \int_0^t \MCK^\eps(s)\ud s$ is of order $\eps^{1-H}$, essentially the same proof applies. Thus, we omit the details here.
\end{proof}

We now introduce the bracket notation $\average{\cdot}$ for averaging with respect to the invariant distribution of fOU process: 
\begin{align*}
\average{g} := \int_\RR g(z)\frac{1}{\sqrt{2\pi}\sigma_{ou}}e^{-\frac{z^2}{2\sigma^2_{ou}}} \ud z = \int_\RR g(\sigma_{ou}z)p(z) \ud z,
\end{align*}
where $p(z)$ is the density of the standard normal distribution, as well as $\overline{\lambda}$ and $\widetilde\lambda$ which will be used throughout the rest of the paper:
\begin{align}
\overline{\lambda} := \sqrt{\average{\lambda^2}}, \quad \widetilde \lambda := \average{\lambda}. \label{def_averagelambda}
\end{align}
Accordingly, we define several important quantities, which are differences between  time averages and  spacial averages:
\begin{align}
& I_t^\eps := \int_0^t \left( \lambda^2(\Yh{s}) - \overline{\lambda}^2\right) \ud s, \label{def_i}\\
&\eta_t^\eps := \int_0^t \left(\lambda(\Yh{s}) - \widetilde\lambda\right) \ud s ,\label{def_eta} \\
&\kappa_t^\eps := \int_0^t \left(\lambda(\Yh{s})\lambda'(\Yh{s}) - \average{\lambda\lambda'}\right) \ud s. \label{def_kappa}
\end{align}
It is proved in Appendix~\ref{app_lemmas} that, by the ergodicity of $\Yh{t}$, these differences are small and of order $\eps^{1-H}$. More properties and estimates regarding $\Yh{t}$ are also stated therein.

Let $C_k$ be the ``probabilists'' Hermite coefficients of the function $\lambda^2(\cdot)$:
\begin{equation*}
C_k :=  \int_{\RR} H_k(z)\lambda^2(\sigma_{ou}z)p(z)\ud z, \quad H_k(z) = (-1)^k e^{z^2/2}\frac{\ud^k \left(e^{-z^2/2}\right)}{\ud z^k}.
\end{equation*}
The Hermite polynomials are naturally associated with  OU processes. Now, we state a further assumption on $\lambda(\cdot)$ which is required in Lemma \ref{lem_momentsw}.
\begin{assump}\label{assump_lambdapower}
There exists $\alpha > 4$ such that 
	\begin{equation*}
	\sum_{k=0}^\infty \frac{\alpha^kC_k^2}{k!} < \infty,
	\end{equation*}
where $C_k$'s are the Hermite coefficients defined above.
\end{assump}

\begin{rem}
	A sufficient condition to Assumption~\ref{assump_lambdapower} given in \cite{GaSo:16} is stated as follows. If $\lambda^2(x)$ is of the form
	\begin{equation*}
	\lambda^2(x) = \int_{-\infty}^{x/\sigma_{ou}} f(y) \ud y,
	\end{equation*}
	where the Fourier transform of the function $f$ satisfies $\abs{\hat f(\nu)} \leq C \exp(-\nu^2)$ for some $C>0$, then Assumption~\ref{assump_lambdapower} is fulfilled. The proof relies on Parseval identity, and we refer to \cite[Lemma A.2]{GaSo:16} for details. 
\end{rem}

In the rest of this section, we study the Merton problem \eqref{def_Vt}, when the stochastic environment is modeled by $\Yh{t}$ with $H$ restricted to $H>\half$, and when the investor's utility is of power type. Note that under such circumstance (in fact, as long as  $H \neq \half$), $\Yh{t}$ is neither a semi-martingale nor a Markov process, thus the usual Hamilton-Jacobi-Bellman partial differential equation is not available. 
However, we have Proposition~\ref{prop_martdistort} which can be applied directly, and this will be the starting point of our derivation of the approximations.

\subsection{First order approximation to the value process}\label{sec_asympVy}
Let $S_t$ follow the dynamics
\begin{align}\label{def_StfastOU}
\ud S_t = S_t\left[\mu(\Yh{t})\ud t  + \sigma(\Yh{t})\ud W_t\right],
\end{align}
where $\Yh{t}$ is the $\eps$-scaled stationary fOU process \eqref{def_Yh} described above with $H > \half$. Then, the wealth process $X_t^\pi$ becomes
\begin{equation}\label{def_XtunderfastY}
\ud X_t^\pi = \pi_t\mu(\Yh{t}) \ud t + \pi_t\sigma(\Yh{t}) \ud W_t.
\end{equation}
Denote by $\Vy_t$ the value process at time $t$ under the current setup:
\begin{equation}
\Vy_t := \esssup_{\pi \in \MCA_t^\eps}\EE\left[U(X_T^\pi)\vert \MCF_t\right],
\end{equation}
where the superscript $\eps$ emphasizes the dependence on $\eps$ brought by $\Yh{t}$, and the notation of admissible set is also changed from $\MCA_t$ to $\MCA_t^\eps$ accordingly.
Directly applying Proposition~\ref{prop_martdistort} with $Y_t = \Yh{t}$ gives the following expression for $\Vy_t$:
\begin{equation}\label{def_Vy}
\Vy_t = \frac{X_t^{1-\gamma}}{1-\gamma}\left[\widetilde \EE\left(e^{\frac{1-\gamma}{2q\gamma}\int_t^T \lambda^2(\Yh{s})\ud s} \Big\vert \MCG_t\right)\right]^q.
\end{equation}

\begin{theo}\label{thm_Vtpowerexpansion}
In the regime of $\eps$ small, under Assumptions~\ref{assump_power} and \ref{assump_lambdapower}, for fixed $t \in [0,T)$, $\Vy_t$ takes the form
\begin{equation}\label{eq_Vtpower}
\Vy_t = Q^\eps_t(X_t) + o(\eps^{1-H}),
\end{equation}
where 
\begin{equation}\label{def_Qeps}
Q_t^\eps(x) = \frac{x^{1-\gamma}}{1-\gamma}e^{\frac{1-\gamma}{2\gamma}\overline\lambda^2(T-t)}\left[1  + \frac{1-\gamma}{\gamma}\left(\phi_t^\eps +  \eps^{1-H}\rho \widetilde\lambda \left(\frac{1-\gamma}{\gamma}\right) \frac{\average{\lambda\lambda'}(T-t)^{H+\half}}{a\Gamma(H + \frac{3}{2})}\right)\right].
\end{equation}
Here $\phi_t^\eps$ is the random process defined as
\begin{align}
\phi_t^\eps = \EE\left[\half\left.\int_t^T \left(\lambda^2(\Yh{s}) - \overline\lambda^2\right) \ud s \right\vert \MCG_t\right],\label{def_phi}
\end{align}
which is of order $\eps^{1-H}$ as proved in Lemma~\ref{lem_moments}\eqref{lem_phi}. The notation $o(\eps^{1-H})$ denotes a $\MCF_t$-adapted random variable whose order is higher than $\eps^{1-H}$ in $L^1$.
\end{theo}

\begin{proof}
In order to obtain \eqref{eq_Vtpower}-\eqref{def_Qeps}, we start by expanding 
\begin{equation}\label{def_Psi}
\Psi_t^\eps:= \widetilde \EE\left[e^{\frac{1-\gamma}{2q\gamma}\int_t^T \left(\lambda^2(\Yh{s})-\overline \lambda^2\right)\ud s}\Big\vert \MCG_t\right],
\end{equation}
then, we apply Taylor formula to the function $x^q$. 

Using the fact that $I_t^\eps$ is ``small'' and Taylor expansion of $e^x$ in $x$, one deduces 
\begin{align}
\Psi_t^\eps
& = \widetilde\EE\left[ 1 + \frac{1-\gamma}{2q\gamma}\int_t^T \left(\lambda^2(\Yh{s})-\overline{\lambda}^2\right)\ud s + R_{[t,T]} \Big\vert\MCG_t \right]\nonumber\\
& = 1 + \frac{1-\gamma}{q\gamma}\widetilde\EE\left[\half\int_t^T \left(\lambda^2(\Yh{s})-\overline{\lambda}^2\right)\ud s \Big\vert\MCG_t \right] + \widetilde \EE\left[ R_{[t,T]} \vert \MCG_t\right]\label{eq_Psi},
\end{align}
where $
R_{[t,T]} = e^{\chi}\left[\frac{1-\gamma}{2q\gamma}\int_t^T \left(\lambda^2(\Yh{s}) - \overline{\lambda}^2\right) \ud s \right]^2$
with $\chi$ being the bounded Lagrange remainder. Thus the term $\widetilde \EE\left[ R_{[t,T]} \vert \MCG_t\right]$ is of order $\eps^{2-2H}$ in $L^1$ by Lemma~\ref{lem_momentsw}\eqref{lem_i}. 

Define the $\widetilde \PP$-martingale $\widetilde \psi_t^\eps$ by
\begin{equation*}
\widetilde\psi_t^\eps = \widetilde\EE\left[\int_0^T G(\Yh{s})\ud s \Big\vert\MCG_t \right],  \quad G(y) = \half(\lambda^2(y) - \overline{\lambda}^2).
\end{equation*}
Taylor expanding $G(\Yh{s})$ at $y  = \Yht{s} := \int_{-\infty}^s \kereps(s-u) \ud \widetilde W_u^Y$, together with $\Yh{s} - \Yht{s} \sim O(\eps^{1-H})$ (see Lemma~\ref{lem_comparison}\eqref{lem_Yhtilde}) yields
\begin{align}
\widetilde\psi_t^\eps 
& = \widetilde \EE\left[ \int_0^T G(\Yht{s}) \ud s \Big\vert\MCG_t\right] + \widetilde \EE\left[\int_0^T G'(\Yht{s})\left(\Yh{s}-\Yht{s}\right) \ud s \Big\vert\MCG_t\right] \nonumber \\
& \quad+ \widetilde \EE\left[\int_0^T G''(\chi_s)\left(\Yh{s}-\Yht{s}\right)^2 \ud s \Big\vert\MCG_t\right]\nonumber \\
& = \widetilde \EE\left[ \int_0^T G(\Yht{s}) \ud s \Big\vert\MCG_t\right] + \widetilde \EE\left[\int_0^T G'(\Yht{s})\int_0^s \rho\left(\frac{1-\gamma}{\gamma}\right)\lambda(\Yh{u}) \kereps(s-u)\ud u \ud s \Big\vert\MCG_t\right] + \MCO(\eps^{2-2H})\nonumber\\
& :=  \widetilde\psi_t^{\eps,1} +  \widetilde\psi_t^{\eps,2}  + \MCO(\eps^{2-2H}). \nonumber
\end{align}
Now it remains to find approximations for $ \widetilde\psi_t^{\eps,j}$, $j=1, 2$, up to order $\eps^{1-H}$. To this end, we need the following estimates in $L^1$:
\begin{align}
&R_t^{(1)} := \eps^{1-H}\int_0^t (T-u)^{H-\half}\left(\lambda(\Yh{u})-\widetilde\lambda\right)\ud u  \sim o(\eps^{1-H}),\label{def_R1}\\
&R_t^{(2)} := \widetilde \EE\left[\int_0^T \left(G'(\Yht{s}) - \average{\lambda\lambda'}\right)\int_0^s \rho\left(\frac{1-\gamma}{\gamma}\right)\lambda(\Yh{u}) \kereps(s-u)\ud u \ud s \Big\vert\MCG_t\right]  \sim o(\eps^{1-H}),\label{def_R2}\\
&R_t^{(3)} := \widetilde \EE\left[\int_0^T \int_0^s \left(\lambda(\Yh{u})-\widetilde\lambda\right) \kereps(s-u)\ud u \ud s \Big\vert \MCG_t\right] \sim o(\eps^{1-H}).\label{def_R3}
\end{align}
The proofs are technical and lengthy, thus deferred to Lemma~\ref{lem_Rj}. To condense the notation, we define
\begin{align}
&\psi_t^\eps = \EE\left[\half\int_0^T \left(\lambda^2(\Yh{s}) - \overline\lambda^2\right) \ud s \Big\vert \MCG_t\right],\label{def_psi}\\
&\vartheta_t^\eps := \int_t^T \EE\left[G'(\Yh{s})\vert\MCG_t\right]\kereps(s-t)\ud s, \label{def_vartheta}\\
&\widetilde\vartheta_t^\eps := \int_t^T \widetilde \EE[G'(\Yht{s})\vert\MCG_t]\kereps(s-t)\ud s, \label{def_varthettilde}
\end{align}
where $\psi_t^\eps$ is a $\PP$-martingale satisfying $\ud \psi_t^\eps = \vartheta_t^\eps \ud W_t^Y$ (see details in Lemma~\ref{lem_moments}\eqref{lem_psi}).  Similarly we have $\ud \widetilde \psi_t^\eps = \widetilde\vartheta_t^\eps \ud \widetilde W_t^Y$, and the difference between $ \vartheta_t^\eps$ and $ \widetilde \vartheta_t^\eps$ is discussed in Lemma~\ref{lem_comparison}\eqref{lem_varthetatilde}).

Next, the terms $ \widetilde\psi_t^{\eps,1}$ and $ \widetilde\psi_t^{\eps,2}$ are computed as follows:
\begin{align*}
\widetilde\psi_t^{\eps,1} &= \widetilde \EE\left[ \int_0^T G(\Yht{s}) \ud s \Big\vert\MCG_t\right]  = \widetilde \EE\left[\int_0^T G(\Yht{s}) \ud s \Big\vert \MCG_0\right] + \int_0^t \widetilde\vartheta_u^\eps \ud \widetilde W_u^Y \hspace{30pt}(\Yht{s}\vert \MCG_0 \stackrel{\MCD}{=} \Yh{s}\vert \MCG_0) \\
& = \EE\left[\int_0^T G(\Yh{s}) \ud s \Big\vert \MCG_0\right] + \int_0^t \vartheta_u^\eps \ud W_u^Y + \int_0^t \left(\widetilde \vartheta_u^\eps - \vartheta_u^\eps\right) \ud W_u^Y - \int_0^t \widetilde \vartheta_u^\eps \rho\left(\frac{1-\gamma}{\gamma}\right)\lambda(\Yh{u})\ud u  \\
& \hspace{220pt} \text{(expression of $\psi_t^\eps$ and $\widetilde \vartheta_u^\eps - \vartheta_u^\eps \sim \MCO(\eps^{2-2H})$)}\\
& = \psi_t^\eps  - \rho\left(\frac{1-\gamma}{\gamma}\right) \int_0^t \vartheta_u^\eps \lambda(\Yh{u}) \ud u +  o(\eps^{1-H}) \hspace{120pt} (\vartheta_u^\eps = \eps^{1-H}\theta_u + \widetilde \theta_u^\eps)\\
& = \psi_t^\eps  - \eps^{1-H}\rho\left(\frac{1-\gamma}{\gamma}\right) \int_0^t \theta_u \lambda(\Yh{u}) \ud u -   \rho\left(\frac{1-\gamma}{\gamma}\right) \int_0^t \widetilde\theta_u^\eps \lambda(\Yh{u}) \ud u + o(\eps^{1-H})\\
& \hspace{345pt} (\widetilde{\theta}_u^\eps \sim o(\eps^{1-H}))\\
\end{align*}
\begin{align*}
\hspace{15pt}& = \psi_t^\eps  - \eps^{1-H}\rho\left(\frac{1-\gamma}{\gamma}\right)\widetilde{\lambda} \int_0^t \theta_u  \ud u -  \eps^{1-H} \rho\left(\frac{1-\gamma}{\gamma}\right) \int_0^t \theta_u \left(\lambda(\Yh{u}) - \widetilde{\lambda}\right) \ud u + o(\eps^{1-H})\\
& \hspace{245pt} \text{(definition of $\theta_u$ and estimate of $R_t^{(1)}$)}\\
& = \psi_t^\eps  - \eps^{1-H}\rho\left(\frac{1-\gamma}{\gamma}\right)\widetilde{\lambda} \frac{\average{\lambda\lambda'}}{a\Gamma(H+\frac{3}{2})}\left(T^{H+\half} - (T-t)^{H+\half}\right) + o(\eps^{1-H}),
\end{align*}
and 
\begin{align*}
\widetilde\psi_t^{\eps,2} &=  \widetilde \EE\left[\int_0^T G'(\Yht{s})\int_0^s \rho\left(\frac{1-\gamma}{\gamma}\right)\lambda(\Yh{u}) \kereps(s-u)\ud u \ud s \Big\vert\MCG_t\right] \\
& = \average{\lambda\lambda'}\widetilde \EE\left[\int_0^T \int_0^s \rho\left(\frac{1-\gamma}{\gamma}\right)\lambda(\Yh{u}) \kereps(s-u)\ud u \ud s \Big\vert\MCG_t\right] + R_t^{(2)} \\
& = \average{\lambda\lambda'}\rho\left(\frac{1-\gamma}{\gamma}\right)\widetilde \lambda \int_0^T \int_0^s  \kereps(s-u)\ud u \ud s  + R_t^{(2)} + R_t^{(3)} \hspace{35pt} \text{(estimates of $R_t^{(2)}$ and $R_t^{(3)}$)}\\
& = \eps^{1-H} \average{\lambda\lambda'}\rho\left(\frac{1-\gamma}{\gamma}\right)\widetilde \lambda \frac{T^{H+\half}}{a\Gamma(H+\frac{3}{2})} + o(\eps^{1-H}). 
\end{align*}
All reasonings are mentioned in the parentheses from line to line and proofs can be found in Lemmas~\ref{lem_moments}\eqref{lem_psi}, \ref{lem_comparison} and \ref{lem_Rj}.
Combining the expansions of $\widetilde\psi_t^{\eps,1}$ and $\widetilde\psi_t^{\eps,2}$ together yields,
\begin{align}\label{eq_psitilde}
\widetilde \psi_t^\eps =  \psi_t^\eps  + \eps^{1-H}\rho\left(\frac{1-\gamma}{\gamma}\right)\widetilde{\lambda} \frac{\average{\lambda\lambda'}}{a\Gamma(H+\frac{3}{2})} (T-t)^{H+\half} + o(\eps^{1-H}).
\end{align}
Subtracting $\int_0^t G(\Yh{u}) \ud u$ from both sides of \eqref{eq_psitilde}, together with \eqref{eq_Psi}, \eqref{def_phi} and \eqref{def_psi}, brings
\begin{align}\label{eq_Psiexpansion}
\Psi_t^\eps = 1 + \frac{1-\gamma}{q\gamma} \left( \phi_t^\eps  + \eps^{1-H}\rho\left(\frac{1-\gamma}{\gamma}\right)\widetilde{\lambda} \frac{\average{\lambda\lambda'}}{a\Gamma(H+\frac{3}{2})} (T-t)^{H+\half} \right) + o(\eps^{1-H}).
\end{align}

Taylor expanding $x^q$ produces the desired result
\begin{align*}
\Vy_t&= \frac{X_t^{1-\gamma}}{1-\gamma}e^{\frac{1-\gamma}{2\gamma}\overline{\lambda}^2(T-t)} \left(\Psi^\eps_t\right)^q \\
& =\frac{X_t^{1-\gamma}}{1-\gamma}e^{\frac{1-\gamma}{2\gamma}\overline{\lambda}^2(T-t)} \left\{1 + \frac{1-\gamma}{\gamma} \left( \phi_t^\eps  + \eps^{1-H}\rho\left(\frac{1-\gamma}{\gamma}\right)\widetilde{\lambda} \frac{\average{\lambda\lambda'}}{a\Gamma(H+\frac{3}{2})} (T-t)^{H+\half} \right)
\right\} + o(\eps^{1-H}).
\end{align*}
Observe that there are two corrections to the leading term: a random component $\phi_t^\eps$, and a deterministic function of $t$, $X_t$ and the spatial average with respect to $\Yh{t}$, both being of order $\eps^{1-H}$.
\end{proof}

\subsection{First order expansion of the optimal strategy}\label{sec_asymppi}
We now turn to the optimal portfolio $\pi^\ast$ that leads to $\Vy_t$. Under the fractional stochastic environment $\Yh{t}$, the form of the optimal strategy \eqref{def_pioptimal} in Proposition~\ref{prop_martdistort} becomes
\begin{equation}\label{def_pioptimalunderfOU}
\pi^\ast_t = \left[\frac{\lambda(\Yh{t})}{\gamma \sigma(\Yh{t})} + \frac{\rho q \xi_t}{\gamma \sigma(\Yh{t})}\right] X_t.
\end{equation}
It is not fully explicit due to the presence of $\xi_t$ given by the martingale representation theorem \eqref{def_xi}. In the regime of $\eps$ small, we use \eqref{eq_Psiexpansion} derived above to obtain the following expansion for $\pi^\ast_t$.
\begin{theo}\label{thm_piexpansion}
Under Assumption~\ref{assump_power} and \ref{assump_lambdapower}, we have the following approximation of the optimal strategy $\pi_t^\ast$:
\begin{align}\label{eq_piapprox}
\pi^\ast_t &=  \left[\frac{\lambda(\Yh{t})}{\gamma \sigma(\Yh{t})} +\eps^{1-H} \frac{\rho(1-\gamma)}{\gamma^2 \sigma(\Yh{t})}\frac{\average{\lambda\lambda'}}{a\Gamma(H+\half)}(T-t)^{H-\half} \right] X_t + o(\eps^{1-H}) \\
 &:= \pi_t^{(0)} + \eps^{1-H} \pi_t^{(1)} + o(\eps^{1-H}). \nonumber
\end{align}
\end{theo}

\begin{proof}
This is done by deriving the expansion of $\xi_t$ from its definition \eqref{def_xi}. We rewrite $M_t$ in terms of $\Psi_t^\eps$ by comparing \eqref{def_Mtmartdistort} to \eqref{def_Psi}, 
\begin{equation*}
M_t = \Psi_t^\eps \; e^{\frac{1-\gamma}{2q\gamma}\int_0^t \lambda^2(\Yh{s})\ud s} \; e^{\frac{1-\gamma}{2q\gamma}\overline{\lambda}^2(T-t)},
\end{equation*}
and then use the approximation \eqref{eq_Psiexpansion} of $\Psi_t^\eps$.

Since, by definition, $M_t$ is a $\widetilde \PP$-martingale, in the following calculation where It\^o's formula is applied to $M_t$, we will only concentrate on the diffusion part. More precisely, the drift terms will not be computed explicitly and are replaced by ``$\ud t$ terms'', in other words, calculations are omitted as long as they do not contribute to the diffusion part:
\begin{align*}
\ud M_t &= M_t (\Psi_t^\eps)^{-1} \ud \Psi_t^\eps + \ud t \text{ terms }  = M_t (\Psi_t^\eps)^{-1} \frac{1-\gamma}{q\gamma}\ud \phi_t^\eps + \ud t \text{ terms }\\
& = M_t (\Psi_t^\eps)^{-1} \frac{1-\gamma}{q\gamma}\ud \psi_t^\eps + \ud t \text{ terms } = M_t (\Psi_t^\eps)^{-1} \frac{1-\gamma}{q\gamma}\vartheta_t^\eps \ud W_t^Y + \ud t \text{ terms } \\
& = M_t (\Psi_t^\eps)^{-1} \frac{1-\gamma}{q\gamma}\vartheta_t^\eps \ud \widetilde W_t^Y. 
\end{align*}
In the above derivation, we have successively used \eqref{eq_Psiexpansion}, $\ud \psi_t^\eps= \ud \phi_t^\eps + \ud t \text{ terms }$, and $\ud \psi_t^\eps= \vartheta_t^\eps \ud W_t^Y$. 

Then $\xi_t$ is easily identified and the approximation is deduced  
\begin{align*}
\xi_t &= (\Psi_t^\eps)^{-1} \frac{1-\gamma}{q\gamma}\vartheta_t^\eps = \eps^{1-H} \frac{1-\gamma}{q\gamma}\theta_t + o(\eps^{1-H})\\
&= \eps^{1-H} \frac{1-\gamma}{q\gamma} \frac{\average{\lambda\lambda'}}{a\Gamma(H+\half)}(T-t)^{H-\half} + o(\eps^{1-H})
\end{align*}
using  $\vartheta_t^\eps = \eps^{1-H}\theta_t + \widetilde \theta_t^\eps$ (see Lemma~\ref{lem_moments}\eqref{lem_psi} for details). Plugging the above expression into \eqref{def_pioptimalunderfOU} yields the desired result \eqref{eq_piapprox}.
\end{proof}

Note that, in the above approximation, both the leading order strategy $\pi_t^{(0)}$ and the first order correction term $\pi_t^{(1)}$ are in feedback forms in terms of the state processes. Therefore, if one decides to track the fast-varying process $\Yh{t}$ to implement $\pi_t^{(0)}$, no further computational cost is required when $\pi_t^{(1)}$ is also included in order to incorporate the inter-temporal hedging. On the other hand, tracking $\Yh{t}$ is not easy and requires sophisticated econometric techniques. This issue will be addressed in Section~\ref{sec_practicalpz}. Before that, we discuss how good the strategy $\pi_t^{(0)}$ is.

\subsection{Asymptotic optimality of $\pi_t^{(0)}$}\label{sec_asymppzpower}
In this subsection, we investigate the relation between $\Vy_t$ and the value function obtained by following the zeroth-order strategy given in \eqref{eq_piapprox}:
\begin{equation*}
\pi_t^{(0)} = \frac{\lambda(\Yh{t})}{\gamma\sigma(\Yh{t})}X_t. \label{def_pzpower}
\end{equation*}
Let $X_t^\pz$ be the wealth process associated to $\pi_t^{(0)}$:
\begin{align*}
\ud X_t^\pz &= \mu(\Yh{t})\pi_t^{(0)} \ud t + \sigma(\Yh{t})\pi_t^{(0)} \ud W_t\\
& = \frac{\lambda^2(\Yh{t})}{\gamma}X_t^\pz \ud t + \frac{\lambda(\Yh{t})}{\gamma}X_t^\pz \ud W_t,
\end{align*}
and denote by $\Vyl_\cdot$ the corresponding value process 
\begin{equation*}
\Vyl_t:= \EE\left[\left.U\left( X_T^\pz \right)\right\vert \MCF_t\right],
\end{equation*}
then, the following result holds:
\begin{cor}\label{cor_optimalityofpz}
Under Assumptions~\ref{assump_power} and \ref{assump_lambdapower}, for fixed $t \in[0,T)$ and the  observed value $X_t$, $\Vyl_t$ is approximated by
	\begin{equation}\label{def_Vylpower}
	\Vyl_t = Q_t^\eps(X_t) + o(\eps^{1-H}),
	\end{equation}
where $Q_t^\eps$ is given in \eqref{def_Qeps}.
\end{cor}
\begin{proof}
In Section~\ref{sec_optimality} Proposition~\ref{prop_Vyl}, such approximation result is given under a more general setup, that is, $U(\cdot)$ is in general form that includes the power utility case \eqref{def_power}. Therefore, the proof here is a straightforward application by adapting the notation $\vz$, $\vo$ ... in Proposition~\ref{prop_Vyl} to the power utility case, and \eqref{def_Vylpower} is easily verified.
\end{proof}

Now, combining Theorem~\ref{thm_Vtpowerexpansion} with Corollary \ref{cor_optimalityofpz} gives that $\Vyl_t- \Vy_t$ is of order $o(\eps^{1-H})$, which indicates that  $\pi_t^{(0)}$
already generates the leading order term plus two corrections of order $\eps^{1-H}$ given by \eqref{def_Qeps}. Therefore, we state that:

\centerline{\it
$\pi_t^{(0)}$ is asymptotically optimal within all admissible strategy $\MCA_t^\eps$ up to order $\eps^{1-H}$.}


\subsection{A practical strategy}\label{sec_practicalpz}
The analysis above relies on the assumption that $\Yh{t}$ is observable or trackable. In other words, to implement the principal term $\pi_t^{(0)}$, one needs to track the fast-varying factor $\Yh{t}$ for any $t \in [0,T]$. This is usually not practical and  long-term investors will not tackle this issue, since it usually  requires high-frequency data and to deal with microstructure issues, as mentioned in \cite{FoSiZa:13}. Instead, they would prefer to look for a practical strategy which does not depend on the factor $\Yh{t}$. To this end, we propose such a strategy and quantify its loss in terms of utility.

In the regime of $\eps$ small, the optimal $Y$-independent strategy proportional to the current wealth level is:
\begin{equation}\label{def_pzbar}
\bar\pi_t^{(0)} = \frac{\overline{\mu}}{\gamma\overline{\sigma}^2} X_t,
\end{equation}
where the coefficients are
\begin{equation*}
\overline{\mu} = \average{\mu}, \quad \overline{\sigma}^2 = \average{\sigma^2}.
\end{equation*}
This is obtained by making the ansatz $\bar\pi_t^{(0)} = c X_t$, and then determining $c$ by optimizing the leading order term of the corresponding problem value. Under self-financing, the wealth process \eqref{def_XtunderfastY} following the ansatz becomes:
\begin{equation*}
X_t^{\bar\pi^{(0)}} = X_0 e^{\int_0^t \left(c\mu(\Yh{s}) - \half c^2\sigma^2(\Yh{s})\right) \ud t + \int_0^t c\sigma(\Yh{s}) \ud W_s},
\end{equation*}
and the value to the problem is computed as
\begin{align*}
V_t^{\bar\pi^{(0)}, \eps} & = \EE[U(X_T^{\bar\pi^{(0)}})\vert \MCF_t] \\
& = \frac{X_t^{1-\gamma}}{1-\gamma}\EE\left(e^{(1-\gamma)\int_t^T \left(c\mu(\Yh{s}) - \half c^2\sigma^2(\Yh{s})\right) \ud t + (1-\gamma)\int_t^T c\sigma(\Yh{s}) \ud W_s}\Big\vert \MCF_t\right) \\
& =  \frac{X_t^{1-\gamma}}{1-\gamma}\widehat \EE\left(e^{\int_t^T \left((1-\gamma)c\mu(\Yh{s}) - \frac{\gamma-\gamma^2}{2} c^2\sigma^2(\Yh{s})\right) \ud t }\Big\vert \MCG_t\right),
\end{align*}
where $W_t - (1-\gamma)c\int_0^t \sigma(\Yh{s})\ud s$ is a standard Brownian motion under $\widehat{\PP}$. Using ergodic property of $\Yh{t}$:
\begin{equation*}
\int_t^T \left(\mu(\Yh{s}) - \overline{\mu}\right) \ud s \sim o(1), \quad \text{and} \quad \int_t^T \left(\sigma^2(\Yh{s})-\overline{\sigma}^2\right) \ud s \sim o(1),
\end{equation*}
and Taylor expanding the function $e^x$ at $x=0$ (a similar derivation as in Theorem~\ref{thm_Vtpowerexpansion}) one deduces:
\begin{align*}
V_t^{\bar\pi^{(0)}, \eps} &=  \frac{X_t^{1-\gamma}}{1-\gamma}e^{[c(1-\gamma)\overline{\mu} - \frac{\gamma-\gamma^2}{2}c^2\overline{\sigma}^2](T-t)}\widehat\EE\left(e^{\int_t^T \left((1-\gamma)c(\mu(\Yh{s})-\overline{\mu}) - \frac{\gamma-\gamma^2}{2} c^2(\sigma^2(\Yh{s})-\overline{\sigma}^2)\right) \ud t }\Big\vert \MCG_t\right) \\
& = \frac{X_t^{1-\gamma}}{1-\gamma}e^{[c(1-\gamma)\overline{\mu} - \frac{\gamma-\gamma^2}{2}c^2\overline{\sigma}^2](T-t)} + o(1).
\end{align*}
The leading order is optimized at $c^\ast = \frac{\overline{\mu}}{\gamma\overline{\sigma}^2}$, which leads to \eqref{def_pzbar}, and gives the optimal leading order term
\begin{equation*}
\frac{X_t^{1-\gamma}}{1-\gamma}e^{\frac{1-\gamma}{2\gamma}\frac{\overline{\mu}^2}{\overline{\sigma}^2}(T-t)}.
\end{equation*}
This can be interpreted as the optimal value with Sharpe ratio $\overline{\mu}/\overline{\sigma}$. 

The loss in utility of using $\bar\pi_t^{(0)}$ is quantified by comparing the above term with the leading order term of $\Vy_t$ given in \eqref{eq_Vtpower}-\eqref{def_Qeps}:
\begin{equation*}
\frac{X_t^{1-\gamma}}{1-\gamma}e^{\frac{1-\gamma}{2\gamma}\overline{\lambda}^2(T-t)},
\end{equation*}
and is measured by the Cauchy-Schwarz gap
\begin{equation*}
\overline{\lambda}^2 = \average{\frac{\mu^2}{\sigma^2}} \geq \frac{\average{\mu^2}}{\average{\sigma^2}}  \geq \frac{\overline{\mu}^2}{\overline{\sigma}^2},
\end{equation*}
as in the Markovian setup in \cite{FoSiZa:13}.
Note that $\overline{\sigma}^2 = \average{\sigma^2}$ is the average which arises in the linear problem of option pricing as observed in   \cite{GaSo:16} in the long-range memory case.

\subsection{Numerical illustration}

Next, we illustrate numerically the asymptotic optimality property of $\pi_t^{(0)}$ and the sub-optimality of $\bar{\pi}_t^{(0)}$, that is, we compute $\Vy_t$, $\Vyl_t$,  and $V_t^{\bar\pi^{(0)}, \eps }$  at time $t=0$ using Monte Carlo simulations, and compare their differences. Using equation \eqref{def_Vy} and changing the measure from $\widetilde \PP$ to $\PP$, one deduces
\begin{align*}
\Vy_0 =  \frac{X_0^{1-\gamma}}{1-\gamma}\left[\EE\left(e^{\left(\frac{1-\gamma}{2\gamma}\right)\int_0^T \lambda^2(\Yh{s})\ud s + \rho\left(\frac{1-\gamma}{\gamma}\right)\int_0^T \lambda(\Yh{s})\ud W_s^Y} \Big\vert \MCG_0\right)\right]^q.
\end{align*} 
 Solving the SDE for $X_t^\pz$ and plugging the solution into the definition of $\Vyl_t$ bring
\begin{align*}
\Vyl_0 = \frac{X_0^{1-\gamma}}{1-\gamma}\EE\left(e^{\left(\frac{-2\gamma^2 + 3\gamma - 1}{2\gamma^2}\right)\int_0^T \lambda^2(\Yh{s})\ud s + \left(\frac{1-\gamma}{\gamma}\right)\int_0^T \lambda(\Yh{s})\ud W_s} \Big\vert \MCF_0\right).
\end{align*}
Similarly, the value process following the $Y$-independent strategy $\bar\pi_t^{(0)}$ is given by
\begin{equation*}
V_0^{\bar\pi^{(0)}, \eps} = \frac{X_0^{1-\gamma}}{1-\gamma} \EE\left(e^{\left(\frac{1-\gamma}{\gamma}\right)\frac{\overline{\mu}}{\overline{\sigma}^2}\int_0^T \mu(\Yh{s})\ud s - \left(\frac{1-\gamma}{2\gamma^2}\right)\frac{\overline{\mu}^2}{\overline{\sigma}^4}\int_0^T \sigma^2(\Yh{s})\ud s  + \left(\frac{1-\gamma}{\gamma}\right)\frac{\overline{\mu}}{\overline{\sigma}^2}\int_0^T \sigma(\Yh{s})\ud W_s} \Big\vert \MCF_0\right).
\end{equation*}

The model parameters are chosen as:
\begin{equation*}
 T = 1, \quad H = 0.6, \quad a = 1, \quad  \gamma = 0.4, \quad \rho = -0.5, \quad \mu(y) = \frac{0.1 \times \lambda(y)}{0.1 + \lambda(y)}, \quad  \lambda^2(y) = \half \int_{-\infty}^{y/\sigma_{ou}}p(z/2)\ud z,
\end{equation*}
where we recall that $p(z)$ is the $\mc{N}(0,1)$-density.
Note that the choice of $\lambda(y)$ above satisfies Assumption~\ref{assump_power}(i) and \ref{assump_lambdapower} (see \cite[Lemma A.2]{GaSo:16}) and $\overline\lambda=\sqrt{\average{\lambda^2}}=0.7$.
Note also that with our choice for $\mu(y)$, both $\mu(y)$ and $\sigma^2(y)=\mu^2(y)/\lambda^2(y)$ are integrable with respect to the invariant distribution of $\Yh{}$, so that $\overline\mu$ and $\overline\sigma^2$ are finite and equal to $.087$ and $.0176$ respectively.

Due to the natural non-Markovian structure, we first generate a ``historical'' path $W_t^Y$ between $-M$ and $0$, and then evaluate each conditional expectation by the average of 500,000 paths. The fast-varying factor $(\Yh{t})_{t \in [0,T]}$ \eqref{eq_Yh} is generated using Euler scheme with mesh size $\Delta t = 10^{-3}$, and $M = (T/\Delta t)^{1.5}$ (cf. \cite{BaLaOpPhTa:03}). 

The numerical results presented in Table \ref{table1} are only for a purpose of illustration as we computed the values for only a few ``omegas" denoted by $\#1, \#2,$ and $\#3$.

\begin{table}[H]
	\centering
	\caption{The value processes $\Vy_0$ \emph{vs.} $\Vyl_0$ \emph{vs.} $V_0^{\bar\pi^{(0)}, \eps}$ for the power utility case.}\label{table1}
	\begin{tabular}{|c|c|c|c|c|}\hline
		 && \#1 & \#2 & \#3   \\ \hline\hline 
		 &$\Vy_0$ & 1.5772 & 1.5644 & 1.3016\\
		$\eps = 1$ & $\Vy_0 - \Vyl_0$& 0.0018 & 0.0019 & 0.0024\\
		& $\Vy_0 - V_0^{\bar\pi^{(0)},\eps}$& 0.0689 & 0.0643 & 0.0820\\ \hline 
		 &$\Vy_0$& 1.5567 & 1.4965 & 1.3183 \\
		 $\eps = 0.5$ & $\Vy_0 -\Vyl_0$& 0.0025 & 0.0028 & 0.0028\\
		 & $\Vy_0 - V_0^{\bar\pi^{(0)},\eps}$& 0.0760 & 0.0593 & 0.0999\\ \hline 		 
		 &$\Vy_0$ & 1.4514 & 1.4417 & 1.3976\\
		 $\eps = 0.1$ & $\Vy_0 - \Vyl_0$ & 0.0026 & 0.0026 & 0.0025\\
		 & $\Vy_0 - V_0^{\bar\pi^{(0)},\eps}$ & 0.0761 & 0.0756 & 0.0823\\ \hline 
		 &$\Vy_0$ & 1.4376& 1.4375 & 1.4105\\
		 $\eps = 0.05$ & $\Vy_0 -\Vyl_0$ & 0.0022 & 0.0022 & 0.0021\\
		 & $\Vy_0 - V_0^{\bar\pi^{(0)},\eps}$ & 0.0750 & 0.0762 & 0.0806\\ \hline 
		 &$\Vy_0$ & 1.4417 & 1.4416 & 1.4276\\
		 $\eps = 0.01$ & $\Vy_0 - \Vyl_0$ & 0.0015 & 0.0015 & 0.0015\\
		 & $\Vy_0 - V_0^{\bar\pi^{(0)},\eps}$ & 0.0724  & 0.0727 & 0.0748\\ \hline 		 

	\end{tabular}
\end{table}

As expected, the strategy $\pi_t^{(0)}$ performs well for $\eps$ small, the relative difference $(\Vy_0 - \Vyl_0)/\Vy_0$ being about $0.1\%$. What is more surprising is that it also performs well even for not so small values of $\eps$. Again, as expected, the sub-optimal ``lazy" strategy $\bar\pi_t^{(0)}$ underperforms $\pi_t^{(0)}$ but it performs relatively well since $(\Vy_0 - V_0^{\bar\pi^{(0)},\eps})/\Vy_0$ is about $5\%$.

\subsection{Comparison with the Markovian case}\label{sec_comparisonMarkov}

In the Markovian case, which corresponds to $H = \half$ in the modeling of $\Yh{t}$ \eqref{def_Yh}, approximations to the value function and the optimal portfolio  have been  derived in  \cite{FoSiZa:13}. They are given by:
\begin{align}
\Vy(t, X_t) &= \frac{X_t^{1-\gamma}}{1-\gamma}e^{\frac{1-\gamma}{2\gamma}\overline{\lambda}^2(T-t)}\left[1 - \sqrt{\eps}\rho\left(\frac{1-\gamma}{\gamma}\right)^2\frac{\average{\lambda\theta'}}{2}(T-t)\right] + \MCO(\eps)\label{def_Vymarkovian}\\
\pi^\ast(t,X_t, \Yh{t}) &= \left[\frac{\lambda(\Yh{t})}{\gamma\sigma(\Yh{t})} + \sqrt{\eps} \frac{\rho(1-\gamma)}{\gamma^2\sigma(\Yh{t})}\frac{\theta'(\Yh{t})}{2}\right]X_t + \MCO(\eps) \label{def_pimarkovian}
\end{align}
where $\theta(y)$ solves the Poisson equation $\half \theta''(y) - ay\theta'(y) = \lambda^2(y) - \overline{\lambda}^2$.	These can be viewed as the limits $\lim_{\eps \to 0}\lim_{H \downarrow \half}$ of our current setup.

However, these limits do not commute. For instance, if we consider the small $\eps$ expansion  of $\pi^\ast$ from \eqref{eq_piapprox} and formally let $H = \half$, we obtain
\begin{equation}\label{def_piotherlimit}
 \left[\frac{\lambda(\Yh{t})}{\gamma \sigma(\Yh{t})} +\sqrt\eps \frac{\rho(1-\gamma)}{\gamma^2 \sigma(\Yh{t})}\frac{\average{\lambda\lambda'}}{a} \right] X_t + o(\sqrt\eps),
\end{equation}
which corresponds to the other order of limits $\lim_{H\downarrow \half}\lim_{\eps \to 0}$. The two expansions \eqref{def_pimarkovian} and \eqref{def_piotherlimit} are different and in particular they track the first order correction in   different ways.

Regarding the value process $\Vy_t$, one first observes that the path-dependent component $\phi^\eps_t$ disappears in \eqref{def_Qeps} in the limit $H \downarrow \half$.
To be precise,
\begin{equation}\label{eq_phivariance}
\lim_{H \downarrow \half}\lim_{\eps \to 0}\eps^{H-1}\phi_t^\eps = 0,
\end{equation}
by Lemma~\ref{lem_moments}\eqref{lem_phi}. This is because $\eps^{H-1}\phi_t^\eps$ converges in distribution to $\mc{N}(0, \sigma_\phi^2(T-t)^{2H})$, where $\sigma_\phi^2$ is given by
\begin{equation*}
\sigma_\phi^2 = \sigma^2_{ou}\average{\lambda\lambda'}^2\left(\frac{1}{\Gamma(2H+1)\sin(\pi H)} - \frac{1}{2H\Gamma^2(H + \half)}\right).
\end{equation*}
Then, the claim \eqref{eq_phivariance} is obtained by setting $H = \half$ in $\sigma_\phi^2$.
Now, we conclude that $\Vy_t$  only exhibits a feedback-type correction when taking a formal limit $H \downarrow \half$ in \eqref{def_Qeps}:
\begin{equation*}
\frac{X_t^{1-\gamma}}{1-\gamma}e^{\frac{1-\gamma}{2\gamma}\overline{\lambda}^2(T-t)}\left[1 + \sqrt{\eps}\rho\left(\frac{1-\gamma}{\gamma}\right)^2\frac{\widetilde \lambda \average{\lambda\lambda'}}{a}(T-t)\right] + o(\sqrt\eps).
\end{equation*}
However the first order correction  is in general not the same as in \eqref{def_Vymarkovian}.

We remark that although the two sets of expansions ($H = \half$ \emph{vs.} $H \in (\half, 1)$) share the same form, the coefficients are not identical. This is because our derivations in Theorem~\ref{thm_Vtpowerexpansion} and \ref{thm_piexpansion} are only valid for $H \in (\half, 1)$, and the singular perturbation is ``singular'' at $H = \half$. Consequently, the order of limits $H \downarrow \half$ and $\eps \to 0$ is not interchangeable, and this leads to different expansion results.

Finally, note that in the Markovian case $H = \half$, the first order correction to $\Vy_t$  is ``deterministic'', while in the case $H > \half$, the stochastic correction $\phi_t^\eps$ of the same order also appears, as a consequence of having long-range dependence in the stochastic environment $\Yh{t}$.

\section{General utilities and fractional stochastic environment}\label{sec_optimality}

In this section, we analyze the nonlinear asset allocation problem using asymptotic methods where the utility function $U(x)$ is general, and when, as in \eqref{def_StfastOU}, the log-return $\mu$ and volatility $\sigma$ of the risky asset $S_t$ are driven by the fast-varying fractional stochastic factor $\Yh{t}$ defined in \eqref{def_Yh} and discussed in Section~\ref{sec_fastfOU}. For the linear pricing problem when the volatility is modeled by $\Yh{t}$, approximation results have been developed in \cite{GaSo:16} using the same technique.

Here, unlike in the power utility case, the representations \eqref{def:Vcorrelated} and \eqref{def_pioptimal}  for the value process and the corresponding optimal strategy are not available, therefore asymptotic expansions can not be done directly. However, we are able to follow the idea developed in \cite[Section~4]{FoHu:16} and \cite[Section~4]{FoHu:17} and partially solve this problem. We first study the value process following a specific strategy called $\pz$ introduced in \eqref{def_pz}, and then show that $\pz$ is the best up to order $\eps^{1-H}$ among the following subset $\widetilde\MCA_t^\eps$ of admissible strategies $\MCA^\eps_t$,
\begin{equation}\label{def_MCAtilde}
\widetilde \MCA_t^\eps[\pzt, \pot, \alpha] := \left\{\pi = \pzt + \eps^\alpha \pot: \pi \in \MCA_t^\eps, \alpha >0, 0< \eps \leq 1 \right\},
\end{equation}
The detailed definition of $\widetilde \MCA_t^\eps$ will be given in Section~\ref{sec_asympoptimality}.
Note that the full optimality of $\pz$ in the whole class $\MCA^\eps_t$ remains an open problem.

In the rest of this section, we briefly review the classical Merton problem, where $\mu$ and $\sigma$ are constants in \eqref{def_St} . Denote by $M(t,x;\lambda)$ the corresponding value function, if the utility $U(x)$ is $C^2(0,\infty)$, strictly increasing, strictly concave, and satisfies the Inada and Asymptotic Elasticity conditions (see \cite{KrSc:03} for details)
\begin{equation*}
U'(0+) = \infty, \quad U'(\infty) = 0, \quad \text{AE}[U] := \lim_{x\rightarrow \infty} x\frac{U'(x)}{U(x)} <1,
\end{equation*}
then, the Merton value function $M(t,x;\lambda)$ is strictly increasing, strictly concave in the wealth variable $x$, and decreasing in the time variable $t$. It is $C^{1,2}([0,T]\times \RR^+)$ and solves the HJB equation
\begin{equation}\label{eq_value}
M_t+\sup_{\pi}\left\{\frac{1}{2}\sigma^2\pi^2M_{xx}+\mu\pi M_x\right\}=
M_t -\frac{1}{2}\lambda^2\frac{M_x^2}{M_{xx}} = 0, \quad M(T,x;\lambda) = U(x),
\end{equation}
where $\lambda = \mu/\sigma$ is the constant Sharpe ratio, and appears as a parameter in \eqref{eq_value}.

Based on the Merton value function $M(t,x;\lambda)$, one defines the risk-tolerance function by
\begin{equation}\label{def_risktolerance}
R(t,x;\lambda) := -\frac{M_x(t,x;\lambda)}{ M_{xx}(t,x;\lambda)}.
\end{equation}
It is clear that $R(t,x;\lambda)$ is continuous and strictly positive due to the regularity, concavity and monotonicity of $M (t,x;\lambda)$. Further properties regarding $R(t,x;\lambda)$ are also discussed in  \cite{KaZa:14},  and \cite{FoHu:16} under general utility with additional assumptions. Some of them are repeatedly used in the derivations and will be mentioned during the proofs.

\subsection{Portfolio performance of a  given strategy}
Denote by $\vz(t,x)$ the value function at ``averaged'' Sharpe-ratio $\overline{\lambda}$
\begin{equation}\label{def_vz}
\vz(t,x) := M(t,x;\overline{\lambda}),
\end{equation}
with $\overline{\lambda}$ given in  \eqref{def_averagelambda}. Using the notations from \cite{FoSiZa:13}:
\begin{align}\label{def_dk}
D_k &:= R(t,x; \overline\lambda)^k \partial_x^k, \qquad k = 1,2, \cdots,\\
\Ltx(\lambda) &:= \partial_t + \frac{1}{2}\lambda^2D_2 + \lambda^2D_1,\label{def_ltx}
\end{align}
and the Merton PDE \eqref{eq_value}, $\vz$ also satisfies
\begin{align}\label{eq_vz}
&\Ltx(\overline\lambda)\vz(t,x) = 0.
\end{align}

The strategy $\pz$ is defined as
\begin{equation}\label{def_pz}
\pz(t,x,y) := -\frac{\lambda(y)}{\sigma(y)}\frac{\vz_x(t,x)}{\vz_{xx}(t,x)} = \frac{\lambda(y)}{\sigma(y)}R(t,x;\overline{\lambda}),
\end{equation}
and our aim is to compute the following quantity:
\begin{equation}\label{def_Vyl}
\Vyl_t := \EE\left[U(X_T^\pz)\vert \MCF_t\right],
\end{equation}
where $X_t^\pz$ is the wealth process following the feedback-form strategy $\pz$
\begin{align}\label{def_Xt}
\ud X_t^\pz &= \mu(\Yh{t})\pz(t,X_t^\pz, \Yh{t}) \ud t + \sigma(\Yh{t})\pz(t, X_t^\pz, \Yh{t}) \ud W_t\\
&= \lambda^2(\Yh{t}) R(t, X_t^\pz; \overline{\lambda}) \ud t + \lambda(\Yh{t})R(t,X_t^\pz; \overline{\lambda}) \ud W_t\nonumber.
\end{align}

The technique used to study $\Vyl_t$ is called ``epsilon-martingale decomposition'', which was firstly introduced in \cite{FoPaSi:00}  to solve the linear pricing problem, and later developed in \cite{FoPaSi:01, GaSo:15, GaSo:16, FoHu:17}. The idea is to make an ansatz $Q_t^{\pz,\eps}$ for $\Vyl_t$ in the form of a martingale plus something small (non-martingale part) with the right terminal condition. Then this ansatz is indeed the approximation to $\Vyl_t$ with an error that is of order of the non-martingale part. Detailed discussion can be found in the references we just mentioned. 

To prove that the ansatz $Q_t^{\pz,\eps}$ is indeed a martingale plus the non-martingale part of the desired order, we further require Assumption~\ref{assump_U} for the utility function and Assumption~\ref{assump_vz} for the value function $\vz(t,x)$. Basically, we work under the same setup of $U(\cdot)$ as in \cite{FoHu:16}, and we restate these requirements here for convenience. Detailed discussions about general utility functions can be found there in Section~2.3.
  \begin{assump}\label{assump_U}
  	Throughout this section, we make the following assumptions on the utility $U(x)$:
  	\begin{enumerate}[(i)]
  		\item\label{assump_Uregularity}  U(x) is $C^6(0,\infty)$, strictly increasing, strictly concave and satisfying the following conditions (Inada and Asymptotic Elasticity):
  		\begin{equation}\label{eq_usualconditions}
  		U'(0+) = \infty, \quad U'(\infty) = 0, \quad \text{AE}[U] := \lim_{x\rightarrow \infty} x\frac{U'(x)}{U(x)} <1.
  		\end{equation}
  		\item\label{assump_Ubddbelow}U(0+) is finite. Without loss of generality, we assume U(0+) = 0.
  		\item\label{assump_Urisktolerance} Denote by $R(x)$ the risk tolerance, 
  		\begin{equation}\label{eq_risktolerance}
  		R(x) := -\frac{U'(x)}{U''(x)}.
  		\end{equation}
  		Assume that $R(0) = 0$, R(x) is strictly increasing and $R'(x) < \infty$ on $[0,\infty)$, and there exists $K\in\RR^+$, such that for $x \geq 0$, and $ 2\leq i \leq 4$,
  		\begin{equation}\label{assump_Uiii}
  		\abs{\partial_x^{(i)}R^i(x)} \leq K.
  		\end{equation}
  		\item\label{assump_Ugrowth} Define the inverse function of the marginal utility $U'(x)$ as $I: \RR^+ \to \RR^+$, $I(y) = U'^{(-1)}(y)$, and assume that, for some positive $\alpha$, $\kappa$, $I(y)$ satisfies the polynomial growth condition:
  		\begin{equation}\label{cond_I}
  		I(y) \leq \alpha + \kappa y^{-\alpha}.
  		\end{equation} 		
  	\end{enumerate}
  \end{assump}
  
Note that the item \emph{(ii)} above excludes the case of power utility $U(x) = \frac{x^{1-\gamma}}{1-\gamma}$ when $\gamma > 1$. However, all results in this section still hold for the case $\gamma >1$, with a slight modification in the proofs. Below is the additional assumption needed jointly on $\vz(t,x)$ and $X_t^\pz$, which is also considered as a hidden assumption on $U(\cdot)$.
  
  \begin{assump}\label{assump_vz} 
 The process $\vz(t,X_t^\pz)$ is in $L^4$ uniformly in $\eps$ and in $t \in [0,T]$, i.e.,
  		\begin{equation}
  		\sup_{t\in[0,T]} \EE\left[ \left(\vz(t,X_t^\pz)\right)^4\right] \leq C_1
  		\end{equation}
  		where $C_1$ is independent of $\eps$.
  \end{assump}

Now we state the following  proposition which gives $Q_t^{\pz,\eps}$.
\begin{prop}\label{prop_Vyl}
	Under Assumption~\ref{assump_power}\eqref{assump_St}-\eqref{assump_lambda}, \ref{assump_lambdapower}, \ref{assump_U} and \ref{assump_vz}, for fixed $t \in [0,T)$, the $\MCF_t$-measurable value process $\Vzl_t$ defined in \eqref{def_Vyl} is approximated by $Q_t^{\pz, \eps}$ up to order $\eps^{1-H}$:
	\begin{equation}\label{eq_Vyl}
	\Vyl_t = Q_t^{\pz, \eps}(X_t^\pz) + o(\eps^{1-H}),
	\end{equation}
	where $Q_t^{\pz, \eps}(x)$ is given by:
	\begin{equation}\label{def_Qt}
	Q_t^{\pz, \eps}(x) = \vz(t,x) +  D_1\vz(t,x) \phi_t^\eps + \eps^{1-H} \rho\widetilde \lambda\vo(t,x).
	\end{equation}
	The function $\vz$ is defined in \eqref{def_vz} and satisfies $\MCL_{t,x}(\overline{\lambda})\vz(t,x) = 0$, $D_1$ and $\widetilde{\lambda}$ are from \eqref{def_dk} and \eqref{def_averagelambda} respectively, $\left(\phi_t^\eps\right)_{t \in [0,T]}$ is the $\MCF_t$-measurable process of order $\eps^{1-H}$ given in \eqref{def_phi}
	and $\vo(t,x)$ is defined as
	\begin{equation}\label{def_vodt}
	\vo(t,x) = D_1^2\vz(t,x) C_{t,T},  \quad C_{t,T} = \frac{\average{\lambda\lambda'}}{a\Gamma(H + \frac{3}{2})}(T-t)^{H+\half}.
	\end{equation}

\end{prop}

\begin{proof}
	Based on the epsilon-martingale decomposition, it suffices to show that $Q_t^{\pz,\eps}$ can be decomposed as $M_t^\eps + R_t^\eps$, where $M_t^\eps$ is a true martingale, and $R_t^\eps$ is of order $o(\eps^{1-H})$. In the sequel, we shall focus on the derivation of determining $Q_t^{\pz,\eps}$, which involves finding corrections  of order $\eps^{1-H}$ so that $R_t^\eps$ is pushed to a higher order, 
	while the proofs regarding $M_t^\eps$ and $R_t^\eps$ are delayed to Appendix~\ref{app_lemmas}. 
	
	Applying It\^{o} formula to $\vz(t,X_t^\pz)$ brings 
	\begin{align}
	\ud \vz(t, X_t^\pz) &= \Ltx(\lambda(\Yh{t}))\vz(t,X_t^\pz) \ud t + \sigma(\Yh{t})\pz(t,X_t^\pz, \Yh{t})\vz_x(t,X_t^\pz) \ud W_t \nonumber\\
	& = \half \left(\lambda^2(\Yh{t}) - \overline{\lambda}^2\right) D_1\vz(t,X_t^\pz) \ud t + \ud M_t^{(1)},\label{eq_s1}
	\end{align}
	where $M_t^{(1)}$ is the martingale given by
	\begin{equation}\label{def_M1}
	\ud M_t^{(1)} = \sigma(\Yh{t})\pz(t,X_t^\pz, \Yh{t})\vz_x(t,X_t^\pz) \ud W_t,
	\end{equation}
	and the relations \eqref{eq_vz} and  $D_1\vz(t,x) = -D_2\vz(t,x)$ have been used.
	
	Recall $\phi_t^\eps$ and $\psi_t^\eps$ defined in \eqref{def_phi} and \eqref{def_psi} respectively, then, we have
	$\ud \psi_t^\eps - \ud \phi_t^\eps = \half\left(\lambda^2(\Yh{t}) - \overline{\lambda}^2\right)\ud t$, and 
	the first term in \eqref{eq_s1} becomes
	\begin{equation*}
	\half \left(\lambda^2(\Yh{t}) - \overline{\lambda}^2\right) D_1\vz(t,X_t^\pz) \ud t = D_1\vz(t,X_t^\pz) \left(\ud \psi_t^\eps - \ud \phi_t^\eps\right).
	\end{equation*}
	To further simplify $D_1\vz(t,X_t^\pz)\ud \phi_t^\eps$, which corresponds to finding the corrector to $\vz(t,X_t^\pz)$ at order $\eps^{1-H}$, we compute the total differential of $D_1\vz(t,X_t^\pz)\phi_t^\eps$ (the arguments of $\vz(t,X_t^\pz)$ will be omitted systematically in the following):
	\begin{align*}
	\ud \left(D_1\vz\phi_t^\eps\right) &= D_1\vz\ud \phi_t^\eps + \phi_t^\eps\Ltx(\lambda(\Yh{t})) D_1\vz \ud t + \phi_t^\eps \sigma(\Yh{t})\pz(t,X_t^\pz,\Yh{t})\partial_x D_1\vz \ud W_t \\
	&\hspace{10pt}+ \sigma(\Yh{t})\pz(t,X_t^\pz, \Yh{t}) \partial_x D_1\vz \ud \average{W,\phi^\eps}_t \\
	&= D_1\vz\ud \phi_t^\eps + \phi_t^\eps\left[\half(\lambda^2(\Yh{t})-\overline\lambda^2)(D_2 + 2D_1)D_1\vz\right]\ud t   \\
   	&\hspace{10pt}+ \phi_t^\eps \lambda(\Yh{t})D_1^2\vz \ud W_t + \rho\lambda(\Yh{t})D_1^2\vz \ud \average{W^Y,\psi^\eps}_t
	\end{align*}
	In the derivation, we have used the definition of $D_1$ and $R(t,x;\lambda)$ (cf. \eqref{def_dk} and \eqref{def_risktolerance}), and
	\begin{equation*}
	\Ltx(\overline{\lambda})D_1\vz = D_1\Ltx(\overline{\lambda})\vz = 0, \text{ and } \ud \average{W, \phi^\eps}_t = \rho\ud \average{W^Y, \psi^\eps}_t.
	\end{equation*}
	The results in Lemma~\ref{lem_moments}\eqref{lem_psi}: $\ud \average{W^Y, \psi^\eps}_t = \vartheta_t^\eps\ud t = \left(\eps^{1-H}\theta_t + \widetilde \theta_t^\eps\right) \ud t$, together with the above derivation produce
	\begin{align}
	\ud \left(D_1\vz\phi_t^\eps\right) &= -\half \left(\lambda^2(\Yh{t}) - \overline{\lambda}^2\right) D_1\vz \ud t +  \phi_t^\eps\left[\half(\lambda^2(\Yh{t})-\overline\lambda^2)(D_2+2D_1)D_1\vz\right]\ud t \nonumber\\
	& \hspace{10pt}+ \eps^{1-H}\rho\lambda(\Yh{t})D_1^2\vz \theta_t\ud t + \rho\lambda(\Yh{t})D_1^2\vz \widetilde\theta_t^\eps \ud t + \ud M_t^{(2)},\label{eq_s2}
	\end{align}
	and 
	\begin{align}\label{def_M2}
	\ud M_t^{(2)} = D_1\vz(t, X_t^\pz) \ud \psi_t^\eps + \phi_t^\eps \lambda(\Yh{t})D_1^2\vz(t, X_t^\pz) \ud W_t.
	\end{align}
	
	The term $\eps^{1-H}\rho\lambda(\Yh{t})D_1^2\vz \theta_t\ud t$ is taken care of by adding the term $\eps^{1-H}\rho\widetilde{\lambda}\vo$ to $Q_t^{\pz, \eps}$. By using the relation $\theta_t = -\partial_t C_{t,T}$, one has
	\begin{align}
	\ud \vo(t,X_t^\pz) &= \Ltx(\lambda(\Yh{t})) \vo(t, X_t^\pz) \ud t + \sigma(\Yh{t})\pz(t,X_t^\pz, \Yh{t})\vo_x(t,X_t^\pz) \ud W_t\nonumber\\
	& = \half(\lambda^2(\Yh{t}) - \overline{\lambda}^2)(D_2 + 2D_1)\vo(t, X_t^\pz) \ud t - D_1^2\vz(t,X_t^\pz)\theta_t\ud t + \ud M_t^{(3)},\label{eq_s3}
	\end{align}
	where $M_t^{(3)}$ is the martingale defined by
	\begin{equation}\label{def_M3}
	\ud M_t^{(3)} = \sigma(\Yh{t})\pz(t,X_t^\pz, \Yh{t})\vo_x(t,X_t^\pz) \ud W_t.
	\end{equation}
	
	Combining equation \eqref{eq_s1}, \eqref{eq_s2} and \eqref{eq_s3} yields
	\begin{align*}
	\ud Q_t^{\pz,\eps}(X_t^\pz) &= \ud \left(\vz(t, X_t^\pz) + D_1\vz(t, X_t^\pz) \phi_t^\eps +  \eps^{1-H}\rho\widetilde{\lambda}\vo(t, X_t^\pz)\right)\\
	& = \phi_t^\eps\left[\half(\lambda^2(\Yh{t})-\overline\lambda^2)(D_2+2D_1)D_1\vz\right]\ud t  +  \eps^{1-H}\rho\left(\lambda(\Yh{t})-\widetilde\lambda\right)D_1^2\vz\theta_t\ud t \\
	& \hspace{10pt}+ \rho\lambda(\Yh{t})D_1^2\vz \widetilde\theta_t^\eps \ud t +  \half\eps^{1-H}\rho\widetilde{\lambda}(\lambda^2(\Yh{t}) - \overline{\lambda}^2)(D_2 + 2D_1)\vo(t, X_t^\pz) \ud t  \\
	& \hspace{10pt}+ \ud M_t^{(1)} + \ud M_t^{(2)} + \eps^{1-H}\rho\widetilde{\lambda}\ud M_t^{(3)}.
	\end{align*}
	Denote by $R_{t,T}^{(j)}$, $j=1, 2, 3, 4$ the first four terms in the above expression
	\begin{align}
	&R_{t,T}^{(1)} := \int_t^T  \phi_s^\eps\left[\half(\lambda^2(\Yh{s})-\overline\lambda^2)(D_2+2D_1)D_1\vz(s,X_s^\pz)\right]\ud s, \label{def_R1general}\\
	&R_{t,T}^{(2)} := \int_t^T  \eps^{1-H}\rho\left(\lambda(\Yh{s})-\widetilde\lambda\right)D_1^2\vz(s, X_s^\pz)\theta_s\ud s,\label{def_R2general}\\
	&R_{t,T}^{(3)} := \int_t^T  \rho\lambda(\Yh{s})D_1^2\vz(s, X_s^\pz) \widetilde\theta_s^\eps \ud s,\label{def_R3general}\\
	&R_{t,T}^{(4)} := \int_t^T  \half\eps^{1-H}\rho\widetilde{\lambda}(\lambda^2(\Yh{s}) - \overline{\lambda}^2)(D_2 + 2D_1)\vo(s, X_s^\pz) \ud s,\label{def_R4general}
	\end{align}
	and it is proved in Lemma~\ref{lem_Rjgeneral} that they are $o(\eps^{1-H})$ terms in $L^1$:
	\begin{equation}\label{eq_Rjgeneral}
	\lim_{\eps\to 0}\eps^{H-1}\; \EE\abs{R_{t,T}^{(j)}} = 0, \quad \forall j = 1, 2, 3, 4.
	\end{equation}
	Lemma~\ref{lem_Mj} also shows that $M_t^{(j)}$, $j=1, 2, 3$ are indeed true $\PP$-martingales.
	
	Therefore, define the martingale $M_t^\eps$ and the non-martingale part $R_t^\eps$ respectively by
	\begin{align*}
	&M_t^\eps := \int_0^t  \ud M_s^{(1)} + \ud M_s^{(2)} + \eps^{1-H}\rho\widetilde{\lambda}\ud M_s^{(3)},\\
	&R_{T}^\eps - R_t^{\eps} := R_{t,T}^{(1)} + R_{t,T}^{(2)} +R_{t,T}^{(3)} +R_{t,T}^{(4)},
	\end{align*}
	and observe that $Q_T^{\pz, \eps}(x) = \vz(T,x) = U(x)$ (since $\phi_T^\eps = \vo(T,x) = 0$ by definition), and then
	we obtain the desired result 
	\begin{align*}
	\Vyl_t &= \EE\left[Q_T^{\pz,\eps}(X_T^\pz)\big\vert \MCF_t\right] = Q_t^{\pz, \eps}(X_t^\pz) + \EE[M_T^\eps - M_t^\eps\vert \MCF_t] + \EE[R_T^\eps - R_t^\eps \vert\MCF_t]\\
	& = Q_t^{\pz, \eps}(X_t^\pz) + \EE[R_{t,T}^{(1)} + R_{t,T}^{(2)} +R_{t,T}^{(3)} +R_{t,T}^{(4)} \vert\MCF_t] = Q_t^{\pz, \eps}(X_t^\pz) + o(\eps^{1-H}).
	\end{align*}
\end{proof}

When the utility $U(\cdot)$ is of power type, the functions $\vz$, $D_1\vz$ and $\vo$ in \eqref{def_Qt} can be computed explicitly:
\begin{align*}
&\vz(t,x) = \frac{x^{1-\gamma}}{1-\gamma}e^{\frac{1-\gamma}{2\gamma}\overline \lambda^2 (T-t)}, \quad D_1\vz(t,x) = \frac{x^{1-\gamma}}{\gamma}e^{\frac{1-\gamma}{2\gamma}\overline \lambda^2 (T-t)}, \\
&\vo(t,x) = \frac{1-\gamma}{\gamma^2}x^{1-\gamma}e^{\frac{1-\gamma}{2\gamma}\overline \lambda^2 (T-t)}\frac{\average{\lambda\lambda'}}{a\Gamma(H + \frac{3}{2})}(T-t)^{H+\half},
\end{align*}
which leads to $Q_t^{\pz,\eps} = Q_t^\eps$ and completes the proof of Corollary~\ref{cor_optimalityofpz}.

\subsection{Asymptotic optimality of $\pz$}\label{sec_asympoptimality}
Now we study the optimality of $\pz$ within the smaller class of admissible strategies $ \widetilde \MCA_t^\eps$, which are of the forms
\begin{equation*}
\widetilde \MCA_t^\eps[\pzt, \pot, \alpha] := \left\{\pi = \pzt + \eps^\alpha \pot: \pi \in \MCA_t^\eps, \alpha >0, 0< \eps \leq 1 \right\}.
\end{equation*}
Note that $\pzt$ and $\pot$ are not required to be feedback controls (even $\pz$ is chosen to be in this form), but only adapted random processes , namely, $\pzt_t \in \MCF_t$ and $\pot_t \in \MCF_t$. Furthermore, we require $\pzt$ and $\pot$ to satisfy Assumption~\ref{assump_piregularity} and \ref{assump_optimality}. The parameter $\alpha$ is restricted to be positive since $\pzt + \delta^0 \pot = \pzt + \pot + \delta^\alpha \cdot 0$. 
To show the optimality of $\pz$, we compare the value processes $V_t^\pz$ to $\Vyp_t$. The later one is defined by
\begin{equation}\label{def_Vzpi}
\Vyp_t := \EE\left[\left.U\left( X_T^\pi \right)\right\vert \MCF_t\right], 
\end{equation}
where $\pi$ denotes an admissible strategy $\pi \in \widetilde \MCA_t^\eps[\pzt, \pot, \alpha]$, and $X_t^\pi$ is the corresponding wealth process:
\begin{equation}\label{def_Xtilde}
\ud X_t^\pi = \mu(\Yh{t})\pi_t \ud t + \sigma(\Yh{t})\pi_t \ud W_t.
\end{equation}
To this end, we first find the approximation of $\Vyp_t$ using the epsilon-martingale decomposition technique as demonstrated in Proposition~\ref{prop_Vyl}, and then asymptotically compare it with \eqref{def_Qt}. 
\begin{assump}\label{assump_piregularity}
	For a fixed choice of $(\pzt$, $\pot$, $\alpha>0)$, we require:
	\begin{enumerate}[(i)] 
		\item The whole family (in $\eps$) of strategies $\{\pzt + \eps^\alpha \pot\}$ is contained in $\MCA^\eps_t$;
		\item The process $\vz(t,X_t^\pi)$ is in $L^4$ uniformly in $\eps$ and $t \in[0,T]$, i.e.,
		\begin{equation}
		\sup_{t\in[0,T]}\EE\left[ \left(\vz(t,X_t^\pi)\right)^4 \right] \leq C_2
		\end{equation}
		where $C_2$ is independent of $\eps$, and $X_t^\pi$ follows \eqref{def_Xtilde} with $\pi = \pzt + \eps^\alpha \pot$.
		
	\end{enumerate}
\end{assump}

\begin{theo}\label{thm_main}
	Under Assumptions~\ref{assump_power}\eqref{assump_St}-\eqref{assump_lambda}, \ref{assump_lambdapower}, \ref{assump_U}, \ref{assump_vz}, \ref{assump_piregularity} and \ref{assump_optimality}, for any family of trading strategies $\widetilde \MCA_t^\eps[\pzt, \pot, \alpha]$, the following limit exists in $L^1$ and satisfies
	\begin{equation}
	\ell := \lim_{\eps \to 0}\frac{\Vyp_t - \Vyl_t}{\eps^{1-H}} \leq 0, \text{ in } L^1,\label{eq_ell}
	\end{equation}
	where $\Vyl_t$ and $\Vyp_t$ are defined in \eqref{def_Vyl} and \eqref{def_Vzpi} respectively.
	
	That is, the strategy $\pz$ given by \eqref{def_pz} which generates $\Vyl_t$ asymptotically outperforms  any family $\widetilde \MCA_t^\eps[\pzt, \pot, \alpha]$ producing $\Vyp_t$ at order $\eps^{1-H}$. Moreover, the inequality can be written according to the following four cases:
	\begin{enumerate}[(i)]
		\item $\pzt = \pz$, $\alpha > (1-H)/2$: $\ell = 0$ and $\Vyp_t = \Vyl_t + o(\eps^{1-H})$; 
		\item $\pzt = \pz$, $\alpha = (1-H)/2$: $-\infty < \ell < 0$ and $\Vyp_t= \Vyl_t + \MCO(\eps^{1-H})$ with $\MCO(\eps^{1-H}) <0$;
		\item $\pzt = \pz$, $\alpha < (1-H)/2$: $\ell = -\infty$ and $\Vyp_t = \Vyl_t + \MCO(\eps^{2\alpha})$ with $\MCO(\eps^{2\alpha}) <0$;
		\item $\pzt \neq \pz$: $\lim_{\eps \to 0} \Vyp_t < \lim_{\eps \to 0} \Vyl_t$,
	\end{enumerate}
	where all relations between $\Vyp_t$ and $\Vyl_t$ hold under $L^1$ sense.
\end{theo}

\begin{rem}
To better understand the limit \eqref{eq_ell}, and to show that different values of $\alpha$ lead to different levels of accuracy in the expansion of $V_t^{\pi, \eps}$, 
the result in Theorem~\ref{thm_main} has been decomposed into the four possible cases. In the cases where we got $\ell = 0$, the result means that the strategy $\pz$ is as good as the family of strategies $\{\pzt + \eps^\alpha\pot\}$ at order $\eps^{1-H}$. In other cases where a strict inequality is obtained, $\pz$ outperforms other strategies, and adding the ``correction'' $\eps^\alpha\pot$ (even if $\pzt = \pz$) will not help in increasing the expected utility of terminal wealth. On the contrary, it leads to a negative effect on the value process $V_t^{\pi,\eps}$ at order $\eps^{2\alpha}$ (\emph{resp.} order one), even when one follows $\pz$ in the leading order (\emph{resp.  $\pzt$ deviate from $\pz$}). Therefore, overall we say that $\pz$ is ``asymptotically'' optimal in the class  $\widetilde\MCA^\eps$ is at least at order $\eps^{1-H}$, no matter what $\alpha$ is.
\end{rem}

\begin{proof}
	We start with the case $\pzt = \pz$. Following the same procedure as in Proposition~\ref{prop_Vyl}, one deduces
	\begin{align*}
	\ud Q_t^{\pz, \eps}(X_t^\pi) &= \ud (\vz(t,X_t^\pz) + D_1\vz(t,X_t^\pz)\phi_t^\eps + \eps^{1-H}\rho\widetilde{\lambda}\vo(t,X_t^\pi))\\
	&  = \ud \widetilde R_t^\eps + \ud \widetilde M_t^\eps + \eps^{2\alpha}\ud N_t^\eps
	\end{align*}
	where $\ud \widetilde M_t^\eps$, $\ud \widetilde R_t^\eps$ and $\ud N_t^\eps$ are given by
	\begin{align*}
	\ud \widetilde M_t^\eps  &= \sigma(\Yh{t})\pi_t \vz_x(t,X_t^\pz) \ud W_t + D_1\vz(t,X_t^\pi) \ud \psi_t^\eps + \phi_t^\eps \sigma(\Yh{t})\pi_t\partial_x D_1\vz(t,X_t^\pi) \ud W_t\\
	& \hspace{10pt}+ \sigma(\Yh{t})\pi_t\vo_x(t,X_t^\pi)\ud W_t,\\
	\ud N_t^\eps & =  \half \sigma^2(\Yh{t})\left(\pot_t\right)^2 \vz_{xx}(t,X_t^\pi)\ud t, \\
	\ud \widetilde R_t^\eps & = \half \phi_t^\eps (\lambda^2(\Yh{t})-\overline{\lambda}^2)(D_2 + 2D_1)D_1^2\vz \ud t + \eps^\alpha \phi_t^\eps \mu(\Yh{t})\pot_t (\partial_x + R(t,X_t^\pi;\overline{\lambda})\partial_{xx})D_1\vz \ud t \\
	&\hspace{10pt}+ \half \eps^{2\alpha} \phi_t^\eps \sigma^2(\Yh{t})(\pot_t)^2\partial_{xx}D_1\vz\ud t + \rho\lambda(\Yh{t})D_1^2\vz \widetilde \theta_t^\eps \ud t + \eps^\alpha \sigma(\Yh{t}) \pot_t \partial_x D_1\vz \vartheta_t^\eps \ud t \\
	&\hspace{10pt} + \eps^{1-H}\rho (\lambda(\Yh{t})-\widetilde{\lambda})D_1^2\vz \theta_t\ud t + \eps^{1-H+\alpha}\rho\widetilde{\lambda}\mu(\Yh{t})\pot_t( \vo_x  + R(t, X_t^\pi; \overline{\lambda})\vo_{xx}) \ud t\\
	& \hspace{10pt} + \half\eps^{1-H}\rho\widetilde \lambda (\lambda^2(\Yh{t})-\overline \lambda^2)(D_2 + 2D_1)\vo \ud t + \half\eps^{1-H + 2\alpha}\rho\widetilde \lambda\sigma^2(\Yh{t})(\pot_t)^2\vo_{xx}\ud t,
	\end{align*}
	and in the expression of $\widetilde R_t^\eps$, the arguments of $\vz(t, X_t^\pi)$ and $\vo(t, X_t^\pi)$  are omitted to condense the notation. 
	
	The process $N_t^\eps$ is strictly decreasing following from the strict concavity of $\vz(t,x) = M(t,x;\overline{\lambda})$. The true martingality of $\widetilde M_t^\eps$ and the fact that $\widetilde R_t^\eps \sim o(\eps^{1-H})$ are guaranteed by Assumption~\ref{assump_optimality}. Thus we deduce
	\begin{align}
	\Vyp_t & = \EE[Q_T^{\pz,\eps}(X_T^\pi)\vert \MCF_t] = Q_t^{\pz, \eps}(X_t^\pi) + \EE[\widetilde R_T^\eps - \widetilde R_t^\eps\vert \MCF_t] + \eps^{2\alpha}\EE[N_T^\eps - N_t^\eps \vert \MCF_t] \nonumber \\
	& = Q_t^{\pz, \eps}(X_t^\pi) + o(\eps^{1-H}) + \MCO(\eps^{2\alpha}) = \Vyl_t +  o(\eps^{1-H}) + \MCO(\eps^{2\alpha}), \label{eq_Vexpansionpz}
	\end{align}
	with $\MCO(\eps^{2\alpha}) <0$. This leads to the first three cases in the theorem.
	
	In the case that $\pzt \neq \pz$, similar derivation brings
	\begin{equation*}
	\ud \vz(t, X_t^\pi) = \ud \widehat R_t^\eps + \ud \widehat M_t^\eps + \ud \widehat N_t^\eps,
	\end{equation*}	
	where $\widehat M_t^\eps$, $\widehat R_t^\eps$ and $\widehat N_t^\eps$ are defined by
	\begin{align*}
	\ud \widehat M_t^\eps &= \sigma(\Yh{t})\pi_t \vz_x(t, X_t^\pi) \ud W_t,\\
	\ud \widehat N_t^\eps &= \half \sigma^2(\Yh{t})\left(\pzt_t - \pz(t, X_t^\pi, \Yh{t})\right)^2\vz_{xx}(t, X_t^\pi) \ud t, \\
	\ud \widehat R_t^\eps &= \half (\lambda^2(\Yh{t})-\overline{\lambda}^2)D_1\vz \ud t + \eps^\alpha \left[\mu\pot_t \vz_x + \sigma^2\pzt_t\pot_t\vz_{xx}  + \half \eps^{\alpha}\sigma^2\left(\pot_t\right)^2\vz_{xx}\right] \ud t,
	\end{align*} 
	and the arguments of $\mu(\Yh{t})$, $\sigma(\Yh{t})$ and $\vz(t, X_t^\pi)$ are omitted in the equation of $\widehat R_t^\eps$. As in the previous case, $\widehat N_t^\eps$ is strictly decreasing due to the concavity of $\vz$, and Assumption~\ref{assump_optimality} ensures that $\widehat M_t^\eps$ is a true martingale and that $\widehat R_t^\eps \sim \MCO(\eps^{(1-H)\wedge \alpha})$. This gives the last case in the theorem, since
	\begin{align}
	V_t^\eps &= \EE[\vz(T, X_T^\pi)\vert \MCF_t] = \vz(t, X_t^\pi) + \EE[\widehat R_T^\eps - \widehat R_t^\eps \vert \MCF_t] + \EE[\widehat N_T^\eps - \widehat N_t^\eps \vert \MCF_t]\nonumber \\
	&  < \vz(t, X_t^\pi) + \MCO(\eps^{(1-H) \wedge \alpha}),\label{eq_Vexpansionpzt}
	\end{align}
	and  $\lim_{\eps \to 0} \Vzl_t = \vz(t, X_t^\pz)$. 
	
\end{proof}

\section{Conclusion}\label{sec_conclusion}

In this paper, we study the nonlinear problem of portfolio optimization in the context of a one-factor fractional stochastic environment. This factor is modeled as a long-range memory fractional Ornstein--Uhlenbeck process with Hurst index in $(\half,1)$ and varying on a fast time-scale characterized by a small parameter $\eps$. In this context, and with power utilities, the value process can be expressed explicitly thanks to a martingale distortion transformation allowing us to perform an expansion as $\eps \to 0$ and obtain explicit formulas for the zeroth order term and the first order corrections of order $\eps^{1-H}$. Likewise, we can expand the optimal strategy and show that its zeroth order approximation is optimal up to the first order in the value process. We also extend this analysis in the case of general utility functions and we show that the asymptotic optimality of the zeroth order strategy in a specific  sub-class of admissible strategies.

\appendix
\section{Technical Lemmas}\label{app_lemmas}
 
In this section, we present several lemmas used in Section~\ref{sec_asymppower} and Section~\ref{sec_optimality}. Note that the constants $K, K'$ in all lemmas do not depend on $\eps$ and may vary from line to line, and we denote the function $G(y)$ as
\begin{equation*}
G(y) = \half(\lambda^2(y) -\overline{\lambda}^2),
\end{equation*}
and $\lVert X\rVert_p := (\EE X^p)^{1/p}$ as the $L^p$-norm of $X$.

\begin{lem}\label{lem_moments} \quad
	\begin{enumerate}[(i)]
		\item\label{lem_psi}
		The martingale $\psi_t^\eps$ defined in \eqref{def_psi}:
		\begin{equation*}
		\psi_t^\eps = \EE\left[\int_0^T G(\Yh{s}) \ud s \Big\vert \MCG_t\right], 
		\end{equation*}
		satisfies
		\begin{equation*}
		\ud \psi_t^\eps = \vartheta_t^\eps \ud W_t^Y, \quad \vartheta_t^\eps := \int_t^T \EE\left[G'(\Yh{s})\vert\MCG_t\right]\kereps(s-t)\ud s. 
		\end{equation*}
		Moreover, the process $\vartheta_t^\eps$ can be written as, for all $t \in [0,T]$
		\begin{equation*}
		\vartheta_t^\eps = \eps^{1-H}\theta_t + \widetilde \theta_t^\eps,
		\end{equation*}
		where $\theta_t$ is a deterministic function
		\begin{equation*}
		\theta_t = \frac{\average{G'}}{a\Gamma(H+\half)}(T-t)^{H-\half},
		\end{equation*}
		and $\widetilde \theta_t^\eps$ is random and of high order than $\eps^{1-H}$ in $L^2$ sense uniformly in $t\in[0,T]$
		\begin{equation*}
		\limsup_{\eps\to 0}\eps^{H-1}\sup_{t\in[0,T]} \ltwonorm{\widetilde \theta_t^\eps} = 0.
		\end{equation*}
		
		\item \label{lem_phi}
		The random component $\phi_t^\eps$ defined in \eqref{def_phi} has the form
		\begin{equation*}
		\phi_t^\eps = \EE\left[\int_t^T G(\Yh{s}) \ud s \Big\vert \MCG_t\right].
		\end{equation*}	
		It is a random variable with mean zero and variance of order $\eps^{2-2H}$:
		\begin{equation*}
		Var(\phi_t^\eps) \leq K\eps^{2-2H}, 
		\end{equation*}
		uniformly in $t \in [0,T]$. Moreover, as $\eps \to 0$, the random variable $\eps^{H-1}\phi_t^\eps$ converges in distribution to $\mc{N}(0, \sigma_\phi^2(T-t)^{2H})$,
		where $\sigma_\phi^2$ is defined by
		\begin{equation*}
		\sigma_\phi^2 = \sigma^2_{ou}\average{\lambda\lambda'}^2\left(\frac{1}{\Gamma(2H+1)\sin(\pi H)} - \frac{1}{2H\Gamma^2(H + \half)}\right).
		\end{equation*}
		\item \label{lem_eta}
		Recall the random process $\eta_t^\eps$ defined in \eqref{def_eta}
		\begin{equation*}
					\eta_t^\eps = \int_0^t \left(\lambda(\Yh{s}) - \widetilde\lambda\right) \ud s ,					
		\end{equation*}
		It is of order $\eps^{1-H}$  in $L^2$ sense uniformly in $t\in[0,T]$:
			\begin{equation*}
			\sup_{t\in[0,T]}\ltwonorm{\eta_t^\eps}\leq K\eps^{1-H}.
			\end{equation*}
		
			\item \label{lem_kappa}
			Recall the random process $\kappa_t^\eps$ defined in \eqref{def_kappa}
			\begin{equation*}
			\kappa_t^\eps = \int_0^t \left(\lambda(\Yh{s})\lambda'(\Yh{s}) - \average{\lambda\lambda'}\right) \ud s. 					
			\end{equation*}
			It is of order $\eps^{1-H}$  in $L^2$ sense uniformly in $t\in[0,T]$:
			\begin{equation*}
			\sup_{t\in[0,T]}\ltwonorm{\kappa_t^\eps} \leq K\eps^{1-H}.
			\end{equation*}			
	\end{enumerate}
\end{lem}

\begin{lem}\label{lem_momentsw}	Under Assumption~\ref{assump_lambdapower},
	\begin{enumerate}[(i)] 
		\item\label{lem_i} The random process $I_t^\eps$ defined in \eqref{def_i}
		\begin{equation*}
		I_t^\eps = \int_0^t \left( \lambda^2(\Yh{s}) - \overline{\lambda}^2\right) \ud s,
		\end{equation*}
		satisfies 
		\begin{equation*}
		\sup_{t\in[0,T]} \EE[(I_t^\eps)^4] \leq K e^{4-4H}.
		\end{equation*}
		
	\item\label{lem_varphi}
	Define the random process $\varphi_t^\eps$ by:
	\begin{equation}\label{def_varphi}
	\varphi_t^\eps = \half \int_0^t \left(\lambda^2(\Yh{s}) - \overline{\lambda}^2\right)\phi_s^\eps \ud s, 
	\end{equation}
	it is of order $o(\eps^{1-H})$ in $L^2$ sense uniformly in $t\in[0,T]$
	\begin{equation*}
	\limsup_{\eps\to 0}\eps^{H-1}\sup_{t\in[0,T]}\ltwonorm{\varphi_t^\eps} = 0.
	\end{equation*}

	\item \label{lem_phiadd}
	The $L^4$ norm of $\phi_t^\eps$ is of order $\eps^{1-H}$, uniformly in $t\in[0,T]$:
		\begin{equation*}
		\sup_{t\in[0,T]}\lfournorm{\phi_t^\eps} \leq K\eps^{1-H}.
		\end{equation*}
	
	\end{enumerate}
\end{lem}

\begin{proof}[Proof of Lemma~\ref{lem_moments} and \ref{lem_momentsw}]
	
	All results are slightly different versions or straightforward generalizations of lemmas in \cite[Appendix~A,B]{GaSo:16}, thus we omit the details here.
\end{proof}

\begin{lem}\label{lem_comparison}\quad
	\begin{enumerate}[(i)]
	\item\label{lem_Yhtilde}
	Denote by $\Yht{t} $ the  $\widetilde \PP$-stationary fractional Ornstein--Ulenbeck process, whose moving average representation is of the form
	\begin{equation*}
	\Yht{t} := \int_{-\infty}^t \kereps(t-s)\ud \widetilde W_s^Y.
	\end{equation*}
	Then, $\sup_{t\in[0,T]}\abs{\Yht{t}-\Yh{t}} \leq K\eps^{1-H}$.
	\item\label{lem_varthetatilde}
	Recall the stochastic process $\vartheta_t^\eps$ defined in \eqref{def_vartheta}, and $\widetilde \vartheta_t^\eps$ defined in \eqref{def_varthettilde}:
	\begin{equation*}
		\vartheta_t^\eps := \int_t^T  \EE[G'(\Yht{s})\vert\MCG_t]\kereps(s-t)\ud s, \quad 
	\widetilde\vartheta_t^\eps := \int_t^T \widetilde \EE[G'(\Yht{s})\vert\MCG_t]\kereps(s-t)\ud s,
	\end{equation*}
	then $\sup_{t\in[0,T]}\abs{\widetilde\vartheta_t^\eps-\vartheta_t^\eps} \leq K\eps^{2-2H}$.
	\end{enumerate}
\end{lem}
\begin{proof}
Part (i) follows straightforwardly by the boundedness of $\lambda(\cdot)$ and the fact that $\MCK(t) - \frac{t^{H-\frac{3}{2}}}{a\Gamma(H-\half)} \in L^1$.

For part (ii), we first compute the conditional distribution of $\Yh{s}$ and $\Yht{s}$given $\MCG_t$, for $t \leq s$:
\begin{align*}
\Yh{s}\vert \MCG_t \stackrel{\PP}{\sim} \mc{N}\left(\int_{-\infty}^t \kereps(s-u) \ud W_u^Y, (\sigma_{0,s-t}^\eps)^2 \right),\quad \Yht{s}\vert \MCG_t \stackrel{\widetilde \PP}{\sim} \mc{N}\left(\int_{-\infty}^t \kereps(s-u) \ud \widetilde W_u^Y, (\sigma_{0,s-t}^\eps)^2 \right),
\end{align*}
with $(\sigma_{l,r}^\eps)^2 = \int_l^r \kereps(u)^2\ud u$.
Therefore the difference is computed as
\begin{align*}
&\widetilde\vartheta_t^\eps-\vartheta_t^\eps \\
& = \int_t^T \left\{\widetilde \EE[G'(\Yht{s})\vert\MCG_t] - \EE[G'(\Yh{s})\vert\MCG_t]\right\}\kereps(s-t)\ud s \\
& = \int_t^T \int_{\RR} \left\{G'\left(\int_{-\infty}^t \kereps(s-u) \ud \widetilde W_u^Y + \sigma_{0,s-t}^\eps z\right) - G'\left(\int_{-\infty}^t \kereps(s-u) \ud W_u^Y + \sigma_{0,s-t}^\eps z\right)\right\}p(z) \ud z\kereps(s-t)\ud s \\
& = - \int_t^T \int_{\RR} G''(\chi) \int_0^t \kereps(s-u)\rho\left(\frac{1-\gamma}{\gamma}\right)\lambda(\Yh{u})\ud u p(z) \ud z\kereps(s-t)\ud s,
\end{align*}
where $\chi$ is a $\MCG_t$-adapted random variable determined by the remainder of Taylor expansion. Now, taking absolute value on both sides, together with fact that $G''$ and $\lambda$ are bounded, and $\int_l^r \kereps(u)\ud u \sim \MCO(\eps^{1-H})$ uniformly in $l, r \in[0,T]$ brings the desired result.
\end{proof}

\begin{lem}\label{lem_Rj} The quantities $R_t^{(j)}$ defined in \eqref{def_R1}-\eqref{def_R3} 
\begin{align*}
&R_t^{(1)} := \eps^{1-H}\int_0^t (T-u)^{H-\half}\left(\lambda(\Yh{u})-\widetilde\lambda\right)\ud u,\\
&R_t^{(2)} := \widetilde \EE\left[\int_0^T \left(G'(\Yht{s}) - \average{\lambda\lambda'}\right)\int_0^s \rho\left(\frac{1-\gamma}{\gamma}\right)\lambda(\Yh{u}) \kereps(s-u)\ud u \ud s \Big\vert\MCG_t\right],\\
&R_t^{(3)} := \widetilde \EE\left[\int_0^T \int_0^s \left(\lambda(\Yh{u})-\widetilde\lambda\right) \kereps(s-u)\ud u \ud s \Big\vert \MCG_t\right],
\end{align*}	
satisfy,  for all $t \in [0,T]$,
\begin{equation}\label{eq_Rj}
\lim_{\eps \to 0} \eps^{H-1}\EE\abs{R_t^{(j)}} = 0, \quad  \forall j = 1, 2, 3,
\end{equation}
\end{lem}
\begin{proof}
	Proof of \eqref{eq_Rj} for $j = 3$. It suffices to prove 
\begin{equation}\label{eq_R3equivalent}
\widetilde R_s^{(3)} := \int_0^s \left(\lambda(\Yh{u}) - \widetilde \lambda\right) \kereps(s-u) \ud u \sim o(\eps^{1-H}) \text{ in }L^2 \text{ uniformly in } s \in[0,T],
\end{equation}
and then use dominated convergence theorem. Noticing that $\MCK(t) - \frac{t^{H-3/2}}{a\Gamma(H - \half)} \in L^1$, it is equivalent to show that 
\begin{equation}\label{def_R3prime}
R_s^{(3')} := \int_0^s \left(\lambda(\Yh{u}) - \widetilde \lambda\right) (s-u)^{H-\frac{3}{2}} \ud u \sim o(1) \text{ in }L^2 \text{ uniformly in } s \in[0,T].
\end{equation}
To this end, we pick a sequence $c_n \to 0$, denote $s_k = (s-c_n)k/N$,  $Z_u^{(3)} = (s-u)^{H-\frac{3}{2}}$, and  recall $\eta_u^\eps$ defined in \eqref{def_eta}, thus
\begin{align*}
R_s^{(3')} &  = \int_0^s Z_u^{(3)} \frac{\ud \eta_u^\eps}{\ud u} \ud u = \int_0^{s-c_n} Z_u^{(3)} \frac{\ud \eta_u^\eps}{\ud u} \ud u  + \int_{s-c_n}^s Z_u^{(3)} \frac{\ud \eta_u^\eps}{\ud u}\ud u \\
& = \sum_{k=0}^{N-1}Z_{s_k}^{(3)}\left(\eta^\eps_{s_{k+1}} - \eta^{\eps}_{s_k}\right) + \sum_{k=0}^{N-1}\int_{s_k}^{s_{k+1}} \left(Z_u^{(3)} - Z_{s_k}^{(3)}\right)\frac{\ud \eta_u^\eps}{\ud u} \ud u+  \int_{s-c_n}^s Z_u^{(3)} \frac{\ud \eta_u^\eps}{\ud u}\ud u \\
&:= R_s^{(3', a)} + R_s^{(3', b)} + R_s^{(3', c)}.   
\end{align*}
The proofs for $R_s^{(3', a)}$  and $R_s^{(3', b)}$ are similar to the ones in \cite[Proposition~4.1 Step1]{GaSo:16}.
By Minkowski's inequality, 
\begin{align*}
\ltwonorm{R_s^{(3',a)}} &\leq 2\sum_{k=0}^{N} \infnorm{Z^{(3)}}\ltwonorm{\eta_{s_k}^\eps} \leq 2(N+1) c_n^{H-\frac{3}{2}}\sup_{u \in [0,s-c_n]}\ltwonorm{\eta_u^\eps} \\
&\leq 2(N+1)c_n^{H-\frac{3}{2}} K\eps^{1-H}.
\end{align*}
The last inequality follows from Lemma~\ref{lem_moments}\eqref{lem_eta}, which implies, for any fixed $N$ and $c_n$, $\ltwonorm{R_s^{(3',a)}}$ goes to 0 uniformly in $s$ as $\eps \to 0$. For the second term
\begin{align*}
\ltwonorm{R_s^{(3',b)}} &\leq \infnorm{\lambda}\sum_{k=0}^{N-1}\int_{s_k}^{s_{k+1}} \left((s-u)^{H-\frac{3}{2}} - (s-s_k)^{H-\frac{3}{2}}\right) \ud u \leq K \sum_{k=0}^{N-1} c_n^{H-\frac{5}{2}} \frac{1}{N^2}\\
& \leq Kc_n^{H-\frac{5}{2}}\frac{1}{N},
\end{align*}
which goes to 0 uniformly in $s$ for any fixed $c_n$, as $N \to 0$. The last term $R_s^{(3',c)}$ also tends to zero as $c_n \to 0$ uniformly in $s$ by Dini's theorem. Therefore, we get the desired result \eqref{eq_Rj} for $j=3$.

The proof of \eqref{eq_Rj} for $j=1$ follows the same routine as in \eqref{def_R3prime} with $Z_u^{(3)}$ replaced by $Z_u^{(1)} = (T-u)^{H-\half}$.

The proof of \eqref{eq_Rj} for $j=2$ is based on the one of \eqref{eq_R3equivalent}. To be specific, one has
\begin{align*}
R_t^{(2)} &= \rho\left(\frac{1-\gamma}{\gamma}\right)\widetilde \lambda\widetilde \EE\left[\int_0^T \left(G'(\Yht{s}) - \average{\lambda\lambda'}\right)\int_0^s  \kereps(s-u)\ud u \ud s \Big\vert\MCG_t\right] \\
&\hspace{10pt}+ \rho\left(\frac{1-\gamma}{\gamma}\right)\widetilde \EE\left[\int_0^T \left(G'(\Yht{s}) - \average{\lambda\lambda'}\right)\widetilde R_s^{(3)} \ud s \Big\vert\MCG_t\right] \\
& \leq K\eps^{1-H}\widetilde\EE\left[\int_0^T \left(G'(\Yht{s}) - \average{\lambda\lambda'}\right)s^{H-\half}\ud s \Big\vert\MCG_t\right] + K'\infnorm{G'}\widetilde\EE\left[\int_0^T \vert\widetilde R_s^{(3)}\vert \ud s \Big\vert\MCG_t\right] \\
& := K\eps^{1-H} \widetilde \EE[R_T^{(2,a)}\vert\MCG_t] + K'R_t^{(2,b)},
\end{align*}
with
\begin{equation*}
R_T^{(2,a)} = \int_0^T \left(G'(\Yht{s}) - \average{\lambda\lambda'}\right)s^{H-\half}\ud s, \quad R_t^{(2,b)} = \widetilde\EE\left[\int_0^T \vert\widetilde R_s^{(3)}\vert \ud s \Big\vert\MCG_t\right] .
\end{equation*}
Now it reduces to show $\widetilde \EE[R_T^{(2,a)}\vert\MCG_t] \to 0$ and $R_t^{(2,b)} \sim o(\eps^{1-H})$ in $L^1$. Using
\begin{align*}
\EE\abs{\widetilde \EE[R_T^{(2,a)}\vert\MCG_t]} \leq K \ltwonorm{R_T^{(2,a)}}, 
\end{align*}
the first one then follows the same line as the proof of \eqref{def_R3prime}. The second one also holds by 
\begin{align*}
\EE\abs{R_t^{(2,b)}} \leq K\left[\EE\int_0^T (\widetilde R_s^{(3)})^2 \ud s \right]^{1/2} \leq K \sup_{s\in[0,T]}\ltwonorm{\widetilde R_s^{(3)}} \sim o(\eps^{1-H})
\end{align*}
and the previously proved result \eqref{eq_R3equivalent}.
\end{proof}

\begin{lem}\label{lem_Mj} The processes $M_t^{(j)}$, $j = 1, 2, 3$ defined in \eqref{def_M1}, \eqref{def_M2} and \eqref{def_M3} are true $\PP$-martingales.
\end{lem}
\begin{proof}
By the Burkholder--Davis--Gundy inequality, it suffices to show $\EE\left[\average{M^{(j)}}_T^{1/2}\right] < \infty$, for $j = 1, 2, 3$.

For the case $j=1$, we compute
\begin{align*}
\ud \average{M^{(1)}}_t = \lambda^2(\Yh{t}) \left(D_1\vz(t, X_t^\pz)\right)^2 \ud t \leq K\left(\vz(t,X_t^\pz)\right)^2 \ud t,
\end{align*}
using Assumption~\ref{assump_power}(i) and the concavity of $\vz$. Then, under Assumption~\ref{assump_vz}
\begin{align*}
\EE\left[\average{M^{(1)}}_T^{1/2}\right] \leq \left[\EE\int_0^T K(\vz(t, X_t^\pz))^2 \ud t\right]^{1/2} \leq  K \sup_{t\in[0,T]}\ltwonorm{\vz(t,X_t^\pz)} < \infty.
\end{align*}

The martingality of $M_t^{(3)}$  is obtained by a similar derivation with additional estimates from \cite[Proposition~3.5]{FoHu:16}:
\begin{equation}\label{eq_Rjestimate}
\abs{R^j(t,x;\overline{\lambda})\partial_x^{(j+1)}R(t,x;\overline{\lambda})} \leq K, \quad 0\leq j \leq 3, \quad \forall (t,x) \in [0,T)\times \RR^+.
\end{equation}
For the case $j=2$, similar reasonings lead to
\begin{align*}
\ud \average{M^{(2)}}_t \leq K[\left(\vartheta_t^\eps\right)^2 + \left(\phi_t^\eps\right)^2] \left(\vz(t, X_t^\pz)\right)^2 \ud t.
\end{align*}
Given Assumption~\ref{assump_vz} for $\vz$ and Lemma~\ref{lem_momentsw}\eqref{lem_phiadd} for $\phi_t^\eps$, it suffices to prove
\begin{equation}\label{eq_vartheta}
\sup_{t\in[0,T]} \lfournorm{\vartheta_t^\eps} < K\eps^{1-H}.
\end{equation}
Recall $\vartheta_t^\eps$ from \eqref{def_vartheta} and use Minkowski inequality, one deduces
\begin{align*}
\EE[\left(\vartheta_t^\eps\right)^4] & \leq \left(\int_t^T \left[\EE\left(\EE[G'(\Yh{s})\vert\MCG_t]\kereps(s-t)\right)^4\right]^{1/4}\ud s \right)^4 = \left(\int_t^T \kereps(s-t) \lfournorm{\EE[G'(\Yh{s})\vert\MCG_t}\ud s \right)^4 \\
& \leq \left(\int_t^T \kereps(s-t) \lfournorm{\left(G'(\Yh{s})\right)}\ud s \right)^4 = \average{\left(\lambda\lambda'\right)^4} \left(\int_t^T \kereps(s-t)\ud s\right)^4,
\end{align*}
and we conclude that $\EE[\left(\vartheta_t^\eps\right)^4] $ is bounded by a constant of order $\eps^{4-4H}$, since
\begin{align*}
\left(\int_0^T \kereps(s)\ud s\right)^4 \leq K\eps^{4-4H}
\end{align*}
using $\MCK(s) - \frac{s^{H-3/2}}{a\Gamma(H-1/2)} \in L^1$. This completes the proof of \eqref{eq_vartheta}, and we get the desired results for $M_t^{(j)}$, $j = 1, 2, 3$.
\end{proof}

\begin{lem}\label{lem_Rjgeneral} The random variable $R_{t,T}^{(j)}$, $j= 1, 2, 3, 4$ defined in \eqref{def_R1general}-\eqref{def_R4general} 
	\begin{align*}
	&R_{t,T}^{(1)} := \int_t^T  \phi_s^\eps\left[\half(\lambda^2(\Yh{s})-\overline\lambda^2)(D_2+2D_1)D_1\vz(s,X_s^\pz)\right]\ud s, \\
	&R_{t,T}^{(2)} := \int_t^T  \eps^{1-H}\rho\left(\lambda(\Yh{s})-\widetilde\lambda\right)D_1^2\vz(s, X_s^\pz)\theta_s\ud s,\\
	&R_{t,T}^{(3)} := \int_t^T  \rho\lambda(\Yh{s})D_1^2\vz(s, X_s^\pz) \widetilde\theta_s^\eps \ud s,\\
	&R_{t,T}^{(4)} := \int_t^T  \half\eps^{1-H}\rho\widetilde{\lambda}(\lambda^2(\Yh{s}) - \overline{\lambda}^2)(D_2 + 2D_1)\vo(s, X_s^\pz) \ud s,
	\end{align*}	
	are of order $o(\eps^{1-H})$:
\begin{equation}\label{app_Rjgeneral}
	\lim_{\eps\to 0}\eps^{H-1}\; \EE\abs{R_{t,T}^{(j)}} = 0, \quad \forall j = 1, 2, 3, 4.
\end{equation}

\end{lem}
\begin{proof}
The proofs here are similar to the ones in Lemma~\ref{lem_Rj}. 

To prove \eqref{app_Rjgeneral} with $j = 1$, we denote $t_k = t + (T-t)k/N$, $Z_s^{(1)} = (D_2 + 2D_1)D_1\vz(s,X_s^\pz)$ and recall $\varphi_t^\eps$ defined in \eqref{def_varphi}, thus $R_{t,T}^{(1)}$ can be written as
\begin{align*}
R_{t,T}^{(1)} &  = \sum_{k=0}^{N-1} \int_{t_k}^{t_{k+1}} Z_s^{(1)} \frac{\ud \varphi_s^\eps}{\ud s} \ud s = \sum_{k=0}^{N-1} \int_{t_k}^{t_{k+1}} Z_{t_k}^{(1)} \frac{\ud \varphi_s^\eps}{\ud s} \ud s + \sum_{k=0}^{N-1} \int_{t_k}^{t_{k+1}} (Z_s^{(1)} - Z_{t_k}^{(1)}) \frac{\ud \varphi_s^\eps}{\ud s} \ud s \\
& = \sum_{k=0}^{N-1} Z_{t_k}^{(1)} (\varphi_{t_{k+1}}^\eps - \varphi_{t_k}^\eps) + \sum_{k=0}^{N-1} \int_{t_k}^{t_{k+1}} (Z_s^{(1)} - Z_{t_k}^{(1)}) \frac{\ud \varphi_s^\eps}{\ud s} \ud s \\
&:= R_{t,T}^{(1, a)} + R_{t,T}^{(1, b)}.
\end{align*}

To proceed the analysis of $R_{t,T}^{(1,a)}$ and $R_{t,T}^{(1,b)}$, we first state two properties of $Z_s^{(1)}$: (a) it has a finite second moment uniformly in $\eps$ and $s\in[0,T]$
\begin{align}\label{eq_Z1moments}
\EE[(Z_s^{(1)})^2] \leq K \EE[(\vz(s, X_s^\pz))^2] \leq K \sup_{s\in[0,T]}\EE[(\vz(s, X_s^\pz))^2] < \infty
\end{align}
using the concavity of $\vz$, the estimates \eqref{eq_Rjestimate} and Assumption~\ref{assump_vz}; and (b) its increments are bounded in $L^2$ by
\begin{equation}\label{eq_bddincrement}
\EE[(Z_u^{(1)}-Z_v^{(1)})^2] \leq K\abs{u-v}.
\end{equation}
Part (b) is obtained by firstly using It\^{o} formula
\begin{align*}
Z_u^{(1)} - Z_v^{(1)} = \int_v^u \Ltx(\lambda(\Yh{s})) Z_s^{(1)}\ud s + \int_v^u\lambda(\Yh{s})D_1 Z_s^{(1)}\ud W_s,
\end{align*}
then, squaring both sides, together with the boundedness of $\lambda$ and the estimates \eqref{eq_Rjestimate}
\begin{align*}
\EE[(Z_u^{(1)} - Z_v^{(1)})^2] \leq K \left(\int_v^u \ltwonorm{\vz(s, X_s^\pz)} \ud s\right)^2 + K' \int_v^u \EE[(\vz(s, X_s^\pz))^2]\ud s,
\end{align*}
and Assumption~\ref{assump_vz}.

Now we proceed to the proof \eqref{app_Rjgeneral} with $j=1$. 
\begin{align*}
\EE\abs{R_{t,T}^{(1,a)}}&\leq \sqrt 2\sum_{k=0}^{N-1} \ltwonorm{Z_{t_k}^{(1)}}[\EE(\varphi_{t_k}^\eps)^2+\EE (\varphi_{t_{k+1}}^\eps)^2]^{1/2} \leq 2N \sup_{s\in[t,T]} \ltwonorm{(Z_{s}^{(1)}} \sup_{s\in[t,T]} \ltwonorm{\varphi_{s}^\eps}
\end{align*}
and is of order $o(\eps^{1-H})$ for any fixed $N$ by Lemma~\ref{lem_momentsw}\eqref{lem_varphi}. For the second term, using \eqref{eq_bddincrement} gives
\begin{align*}
\EE\abs{R_{t,T}^{(1,b)}} 
&\leq \infnorm{\lambda}\sum_{k=0}^{N-1}\int_{t_k}^{t_{k+1}} \ltwonorm{Z_s^{(1)}-Z_{t_k}^{(1)}}\ltwonorm{\phi_s^\eps} \ud s \\
& \leq \infnorm{\lambda}K\eps^{1-H}\sum_{k=0}^{N-1}\int_{t_k}^{t_{k+1}} (s-t_k)^{1/2} \ud s = K\eps^{1-H}\frac{1}{\sqrt N},
\end{align*}
and 
\begin{align*}
	\lim_{\eps\to 0}\eps^{H-1}\; \EE\abs{R_{t,T}^{(1,b)}} \leq \frac{K}{\sqrt N}
\end{align*}
holds for any $N$. Thus, we have the desired result, by letting $N \to \infty$.

Proof of \eqref{app_Rjgeneral} for $j=2$ (\emph{resp.} $j=4$) is done by applying essentially the same argument to $Z_s^{(2)} = D_1^2\vz(s, X_s^\pz)\theta_s$ (\emph{resp.} $Z_s^{(4)} = (D_2 + 2D_1)\vo(s, X_s^\pz)$) which also satisfies \eqref{eq_Z1moments} and \eqref{eq_bddincrement}, and $\eps^{1-H}\eta_t^\eps$ (\emph{resp.} $\eps^{1-H}I_t^\eps$).

Proof of \eqref{app_Rjgeneral} for $j=3$ is given by
\begin{align*}
\EE\abs{R_{t,T}^{(3)}} &\leq \rho \int_t^T \EE\abs{\lambda(\Yh{s})D_1^2\vz(s,X_s^\pz)\widetilde{\theta}_s^\eps} \ud s  \leq K\int_t^T \ltwonorm{\vz(s,X_s^\pz)}\ltwonorm{\widetilde{\theta}_s^\eps} \ud s \\
& \leq K \sup_{s\in[t,T]}\ltwonorm{\vz(s,X_s^\pz)} \sup_{s\in[t,T]}\ltwonorm{\widetilde{\theta}_s^\eps}
\end{align*}
and Lemma~\ref{lem_moments}\eqref{lem_psi}.
\end{proof}

\section{Assumptions for Theorem~\ref{thm_main}}\label{appendix_addasump}

This set of assumptions is used to establish the approximation accuracy \eqref{eq_Vexpansionpz} (\emph{resp.} \eqref{eq_Vexpansionpzt}) to $\Vyp_t$ defined in \eqref{def_Vzpi}. To be specific, these assumptions will ensure that $\widetilde M_t^\eps$ (\emph{resp.} $\widehat M_t^\eps$) is a true martingale and that $\widetilde R_t^{\eps}$ (\emph{resp.} $\widehat R_t^\eps$)  is of order $o(\eps^{1-H})$ (\emph{resp.} $\MCO(\eps^{(1-H) \wedge \alpha})$).
\begin{assump}\label{assump_optimality}
	Let $\widetilde\MCA_t^\eps\left[\pzt,\pot,\alpha\right]$ be the family of trading strategies defined in \eqref{def_MCAtilde}. Recall that $X^\pi$ is the wealth process generated by the strategy $\pi=\pzt+\eps^\alpha\pot$ as defined in \eqref{def_Xtilde}. In order to condense the notation,  we systematically omit  the arguments $(s,X_s^\pi)$ of $\vz$ and $\vo$ and the argument $\Yh{s}$ of $\mu$ and $\sigma$ in what follows. According to the different cases, we further require:
	\begin{enumerate}[(i)]
		\item\label{assump_optimality_eq} If $\pzt \equiv \pz$, the quantities below, for any $t \in [0,T]$,  are of order $\eps^{1-H}$ in $L^1$ sense:
		
		$\int_t^T \phi_s^\eps \mu \pot_t (\partial_x + R(s,X_s^\pi;\overline{\lambda})\partial_{xx})D_1\vz \ud s$,
		$\int_t^T \phi_s^\eps \sigma^2 (\pot_t)^2\partial_{xx}D_1\vz\ud s$,
		
		and the following quantities are uniformly bounded in $\eps$:
		
		$\EE\int_0^T \left(\sigma\pot_t\vz_x\right)^2 \ud s $, 
		$\EE \abs{\int_0^T  \mu\pot_t \vo_x \ud s}$,	
		$\EE \abs{\int_0^T  \mu\pot_t R(t, X_t^\pi; \overline{\lambda}) \vo_{xx} \ud s}$,		
		$\EE \abs{\int_0^T  \sigma^2\left(\pot_t\right)^2 \vo_{xx} \ud s}$,
		
		$\EE\left(\int_0^T \left(\sigma\pot_t\vz_x\phi_s^\eps\right)^2 \ud s\right)^\half$, 		
		$\EE\left(\int_0^T \left(\sigma\pot_t\vo_x\right)^2 \ud s\right)^\half$,
		
		\item\label{assump_optimality_neq} If $\pzt \not\equiv \pz$, we require the uniformly boundedness (in $\eps$) of the following:
		
		$\EE\abs{\int_0^T  \mu\pot_t \vz_x\ud s}$,
		$\EE\abs{\int_0^T  \sigma^2\pzt_t\pot_t \vz_{xx}\ud s}$,
		$\EE\abs{\int_0^T  \sigma^2 \left(\pot_t\right)^2\vz_{xx}\ud s}$,
		$\EE\left(\int_0^T \left(\sigma\pzt_t\vz_x\right)^2 \ud s\right)^\half$,
		
		$\EE\left(\int_0^T \left(\sigma\pot_t\vz_x\right)^2 \ud s\right)^\half$.
	\end{enumerate}
\end{assump}
\bibliographystyle{plainnat}
\bibliography{Reference}

\end{document}